\documentclass[review]{elsarticle}

\usepackage{lineno,hyperref}
\modulolinenumbers[5]

\journal{Journal of Computer and System Sciences}

\usepackage[utf8]{inputenc}
\usepackage{mathtools}
\usepackage{amssymb,microtype,comment,enumitem}
\usepackage{amsthm}
\usepackage{amsmath}
\usepackage{amsfonts}

\newtheorem{theorem}{Theorem}
\newtheorem{lemma}{Lemma}
\newtheorem{example}{Example}
\newtheorem{definition}{Definition}
\newtheorem{proposition}{Proposition}
\newtheorem{corollary}{Corollary}
\newtheorem{fact}{Fact}

\newcommand{\Num}{\mathit{Num}}
\newcommand{\thefirstsort}{sort $\bf v$}
\newcommand{\thesecondsort}{sort $\bf n$}

\newcommand{\ttl}{{\tt l}}

\newcommand{\tts}{{\tt s}}

\newcommand{\ttu}{{\tt u}}

\newcommand{\ttx}{{\tt x}}
\newcommand{\tty}{{\tt y}}
\newcommand{\ttz}{{\tt z}}
\newcommand{\idx}{\mathit{index}}
\newcommand{\polylog}{\dpolylog}
\newcommand{\dpolylog}{\mathrm{PolylogTime}}
\newcommand{\polylogspace}{\mathrm{PolylogSpace}}
\newcommand{\logspace}{\mathrm{L}}
\newcommand{\nlogspace}{\mathrm{NL}}
\newcommand{\pspace}{\mathrm{PSPACE}}
\newcommand{\mA}{\mathbf{A}}
\newcommand{\mB}{\mathbf{B}}
\newcommand{\cC}{\mathfrak{C}}
\newcommand{\N}{\mathbb{N}}
\newcommand{\enc}{\mathrm{bin}}
\newcommand{\ins}{\mathrm{Ins}}
\newcommand{\IFP}[2]{\ensuremath{\left[\mathrm{IFP}_{#1}{#2}\right]}}
\newcommand{\IFPx}[1]{\IFP{\bar{\tt x},X}{#1}}
\newcommand{\val}{\mathit{val}}
\newcommand{\en}{\hat{n}}
\newcommand{\STATE}{\mathrm{STATE}}
\newcommand{\ENDMARKS}{\mathrm{ENDMARKS}}
\newcommand{\HEAD}{\mathrm{HEAD}}
\newcommand{\ONE}{\mathrm{ONE}}
\newcommand{\BLANK}{\mathrm{BLANK}}
\newcommand{\CONFIG}{\mathrm{CONFIG}}
\newcommand{\bit}{\mathrm{BIT}}

\newcommand{\var}{\mathrm{VAR}}

\begin{document}

\begin{frontmatter}

\title{Descriptive Complexity of Deterministic Polylogarithmic Time and Space\tnoteref{thanks}}
\tnotetext[thanks]{The research reported in this paper results from the project {\em Higher-Order Logics and Structures} supported by the Austrian Science Fund (FWF: \textbf{[I2420-N31]}) and the Research Foundation Flanders (FWO:\textbf{[G0G6516N]}). 
It was further supported by the the Austrian Ministry for Transport, Innovation and Technology, the Federal Ministry of Science, Research and Economy, and the Province of Upper Austria in the frame of the COMET center SCCH.}

\author[addr1]{Flavio Ferrarotti\corref{cor1}}
\ead{flavio.ferrarotti@scch.at}
\author[addr1]{Sen\'en Gonz\'alez}
\author[addr2]{Jos\'e Mar\'{\i}a Turull Torres}
\author[addr3]{Jan Van den Bussche}
\author[addr3]{Jonni Virtema}

\cortext[cor1]{Corresponding author}
\address[addr1]{Software Competence Center Hagenberg, Austria}
\address[addr2]{Universidad Nacional de La Matanza, Argentina}
\address[addr3]{Hasselt University, Belgium}

\begin{abstract}
We propose logical characterizations of problems solvable in
deterministic polylogarithmic time ($\dpolylog$) and polylogarithmic space ($\polylogspace$).  We
introduce a novel two-sorted logic that separates the elements
of the input domain from the bit positions needed to address
these elements. We prove that the inflationary and partial fixed point vartiants of this logic  capture $\dpolylog$ and $\polylogspace$, respectively. In the course of proving that our logic indeed
captures $\dpolylog$ on finite ordered structures, we introduce
a variant of random-access Turing machines that can access the
relations and functions of a structure directly.  We
investigate whether an explicit predicate for the ordering of
the domain is needed in our $\dpolylog$ logic. Finally, we present the open 
problem of finding an exact characterization of order-invariant
queries in $\dpolylog$.
\end{abstract}


\end{frontmatter}


\section{Introduction}

The research area known as Descriptive Complexity
\cite{fmta_book,gurevich_complexity,immerman_book} relates
computational complexity to logic.  For a complexity class of
interest, one tries to come up with a natural logic such that a
property of inputs can be expressed in the logic if and only if
the problem of checking the property belongs to the complexity
class.  An exemplary result in this vein is that a family $\cal
F$ of finite structures (over some fixed finite vocabulary) is
definable in existential second-order logic (ESO), if and only if
the membership problem for $\cal F$ belongs to NP
\cite{fagin_theorem}.  We also say that ESO \emph{captures} NP\@.
The complexity class P is captured, on ordered finite
structures, by a \emph{fixed point logic}:
the extensions of first-order logic with
least fixed points~\cite{imm_relpol,vardi_comp}.

After these two seminal results, many more capturing results have
been developed, and the benefits of this enterprise has been
well articulated by several authors in the references given
earlier, and others \cite{ahv_book}.  We just mention here the
advantage of being able to specify properties of structures (e.g., data
structures and databases) in a logical, declarative manner; at the
same time, we are guaranteed that our computational power is well
delineated.

The focus of the present paper is on computations taking deterministic 
polylogarithmic time, i.e., time proportional to $(\log n)^k$ for
some arbitrary but fixed $k$.  Such computations are practically
relevant and common on ordered structures.  Well known examples are binary
search in an array or search in a balanced search tree.  Another
natural example is the computation of $f(x_1,\dots,x_r)$, where
$x_1$, \dots, $x_r$ are numbers taken from the input structure
and $f$ is a function computable in polynomial time when numbers
are represented in binary.

Computations with sublinear time complexity can be formalized
in terms of Turing machines with random access to the input
\cite{immerman_book}.  When a family $\cal F$ of ordered finite structures
over some fixed finite vocabulary is defined by some deterministic 
polylogarithmic-time random-access Turing machine, we say that 
$\cal F$ belongs to the complexity class $\dpolylog$.
In this paper, we show how this complexity class can be captured by
a new logic which we call \emph{index logic}.

Index logic is two-sorted; variables of the first sort range over
the domain of the input structure.  Variables of the second sort
range over an initial segment of the natural numbers; this
segment is bounded by the logarithm of the size of the input
structure.  Thus, the elements of the second sort represent the
bit positions needed to address elements of the first sort.
Index logic includes full fixed point logic on the second sort.
Quantification over the first sort, however, is heavily
restricted.  Specifically, a variable of the first sort can only
be bound using an address specified by a subformula that defines
the positions of the bits of the address that are set.
This ``indexing mechanism'' lends index logic its name.

In the course of proving our capturing result, we consider a new
variant of random-access Turing machines.  In the standard
variant, the entire input structure is presented as one binary
string.  In our new variant, the different relations and
functions of the structure can be accessed directly.  We will
show that both variants are equivalent, in the sense that they
lead to the same notion of $\dpolylog$.  We note that, in descriptive
complexity, it is a common practice to work only with relational
structures, as functions can be identified with their
graphs.  In a sublinear-time setting, however, this does not
work.  Indeed, let $f$ be a function and denote its graph by
$\tilde f$.  If we want to know the value of $f(x)$, we cannot spend the linear
time needed to find a $y$ such that $\tilde f(x,y)$ holds.  Thus,
in this work, we allow structures containing functions as well as
relations.


We also devote attention to gaining a detailed understanding of
the expressivity of index logic.  Specifically, we observe that
order comparisons between quantified variables of the first sort
can be expressed in terms of their addresses.  For constants of
the first sort that are directly given by the structure, however,
we show that this is not possible.  In other words, index logic
without an explicit order predicate on the first sort would no
longer capture $\dpolylog$ for structures with constants. 

Finally, we introduce a variant of index logic with partial fixed point operators and show that it captures $\polylogspace$. This result is analogous to the classical result regarding the descriptive complexity of PSPACE, which is captured over ordered structures by first-order logic with the addition of partial fixed point operators~\cite{Vardi82}. For consistency, we define $\polylogspace$ using the model of direct-access Turing machines, i.e., the variant of the random-access Turing machine that we introduce in this paper. As with $\dpolylog$, both models of computation lead to the same notion of $\polylogspace$. Moreover, we show that, in the case of $\polylogspace$, random-access to the input-tape can be replaced with sequential-access without having any impact on the complexity class. Similar to PSPACE, the nondeterministic and deterministic $\polylogspace$ classes coincide. It is interesting to note that beyond the problems in nondeterministic logarithmic space, there are well known natural problems that belong to $\polylogspace$ (see examples below, under related work).   

A preliminary version of this paper was presented at the 26th International Workshop in Logic, Language, Information, and Computation~\cite{FerrarottiGTBV19}. This is an extended improved version which in addition to the full proofs of the results on deterministic polylogarithmic time reported in~\cite{FerrarottiGTBV19}, also considers polylogarithmic space and its descriptive characterization in terms of a variant of index logic.

\paragraph{Related work} 
Many natural fixed point computations, such as transitive closure, converge after a polylogarithmic number of steps. This motivated the study in~\cite{GroheP2017} of a fragment of fixed point logic with counting (FPC) that only allows polylogarithmically many iterations of the fixed point operators (\textsc{polylog}-FPC). They noted that on ordered structures \textsc{polylog}-FPC captures NC, i.e., the class of problems solvable in parallel polylogarithmic time. This holds even in the absence of counting, which on ordered structures can be simulated using fixed point operators. Moreover, an old result in~\cite{Immerman81} directly implies that \textsc{polylog}-FPC is strictly weaker than FPC with regards to expressive power.

It is well known that the (nondeterministic) logarithmic time hierarchy corresponds exactly to the set of first-order definable Boolean queries (see~\cite{immerman_book}, Theorem~5.30). The relationship between uniform families of circuits within {NC}$^1$ and nondeterministic random-access logarithmic time machines was studied in~\cite{barrington:jcss1990}. 
However, the study of descriptive complexity of classes of problems decidable by \emph{deterministic} formal models of computation in polylogarithmic time, i.e., the central topic of this paper, has been overlooked by previous works.

On the other hand, \emph{nondeterministic} polylogarithmic time complexity classes, defined in terms of alternating random-access Turing machines and related families of circuits, have received some attention~\cite{Barr92,FerrarottiGST18}.  
 Recently, a theorem analogous to Fagin's famous theorem~\cite{fagin_theorem}, was proven for nondeterministic polylogarithmic time~\cite{FerrarottiGST18}. For this task, a restricted second-order logic for finite structures, where second-order quantification ranges over relations of size at most polylogarithmic in the size of the structure, and where first-order universal quantification is bounded to those relations, was exploited. This latter work, is closely related to the work on constant depth quasi-polynomial size AND/OR circuits and the corresponding restricted second-order logic in~\cite{Barr92}. Both logics capture the full alternating polylogarithmic time hierarchy, but the additional restriction in the first-order universal quantification in the second-order logic defined in~\cite{FerrarottiGST18}, enables a one-to-one correspondence between the levels of the polylogarithmic time hierarchy and the prenex fragments of the logic, in the style of a result of Stockmeyer~\cite{Stockmeyer76} regarding the polynomial-time hierarchy. Unlike the classical results of Fagin and Stockmeyer~\cite{fagin_theorem,Stockmeyer76}, the results on the descriptive complexity of nondeterministic polylogarithmic time classes only hold over ordered structures.

Up to the authors knowledge, very little is known regarding the relationship of $\polylogspace$ with the main classical complexity classes (see~\cite{papadimitriou} and~\cite{gj_intract}). As usual, let $\logspace$ and $\nlogspace$ denote deterministic and nondeterministic logarithmic space, respectively. Further, let $\logspace^j$ denote $\mathrm{DSPACE}[(\left\lceil\log n \right\rceil)^{j}]$.
The following relations are known:
\begin{enumerate}[label=(\roman*)]
 \item $\polylogspace \neq \mathrm{P}$, and it is \emph{unknown} whether  $\polylogspace \subseteq \mathrm{P}$.
 \item $\mathrm{PolylogSpace} \neq \mathrm{NP}$, and it is \emph{unknown} whether  $\mathrm{PolylogSpace} \subseteq \mathrm{NP}$.
 \item Obviously: $\logspace \subseteq \nlogspace \subseteq \logspace^2 \subseteq \polylogspace \subseteq \mathrm{DTIME}[2^{(\left\lceil\log n \right\rceil)^{O(1)}}]$, the latter class being known as quasi-polynomial time ($\mathrm{QuasiP}$).
\item\label{item4} For all $i \geq j \geq 1$, $\logspace^j$ uniform $\mathrm{NC^{i}}$ $\subseteq \logspace^i$ (see~\cite{Borodin77}); hence we have that $\polylogspace$ $\mathrm{uniform}$ $\mathrm{NC}$ $\subseteq \polylogspace$.
\item  For all $i \geq 1$, let $\mathrm{SC}^{i} := \mathrm{DTIME{-}DSPACE}(n^{O(1)}, (\log n)^{i})$ and let $\mathrm{SC} := \bigcup_{i \in \mathbb{N}}\mathrm{SC^{i}}$ (see~\cite{Greenlaw95}). It follows that $\polylogspace = \mathrm{SC} \cap \mathrm{P}$.
\end{enumerate}

Some interesting natural problems in $\polylogspace$ which are not known to be in $\nlogspace$ follow. By item~\ref{item4} above, we get that division, exponentiation, iterated multiplication of integers~\cite{Reif86}, and integer matrix operations, such as exponentiation, computation of the determinant, rank and the characteristic polynomial (see~\citep{MateraT1997} and~\cite{GrossoHMST2000} for detailed algorithms in $\logspace^2$), are all in $\polylogspace$. Other well-known problems in the class are $k$-colorability of graphs of bounded tree-width~\cite{GottlobLS02}, primality, 3NF test, BCNF test for relational schemas of bounded tree-width~\cite{GottlobPW06,GottlobPW10}, and the circuit value problem of only EXOR gates~\cite{papadimitriou}.
Finally, in~\cite{BeaudryM95} an interesting family of problems is presented. It is shown that, for every $k \geq 1$, there is an algebra $(S; +, .)$ over matrices such that the depth $O(\log n)^{k}$ straight linear formula problem over $M(S; +, .)$ is $\mathrm{NC}^{k+1}$ complete under $\logspace$ reducibility. Now, by~\ref{item4} above, these problems are in $\mathrm{DSPACE}[(\log n)^{k+1}]$.

\section{Preliminaries}\label{sec:preliminaries}
	

We allow structures containing functions as well as relations and
constants. Unless otherwise stated, we work with finite ordered
structures of finite vocabularies. A finite structure $\bf A$ of
vocabulary \[\sigma = \{R^{r_1}_1, \ldots, R^{r_p}_p, c_1, \ldots
c_q, f^{k_1}_1, \ldots, f^{k_s}_s\},\] where each $R^{r_i}_i$ is
an $r_i$-ary relation symbol, each $c_i$ is a constant symbol, and
each $f^{k_i}_i$ is a $k_i$-ary function symbol, consists of a
finite domain $A$ and interpretations for all relation, constant,
and function symbols in $\sigma$. An interpretation of a symbol
$R^{r_i}_i$ is a relation $R^{\bf A}_i \subseteq A^{r_i}$, of a
symbol $c_i$ is a value $c_i^{\bf A} \in A$, and of a symbol
$f^{k_i}_i$ is a function $f^{\bf A}_i: A^{k_i} \rightarrow A$.
A finite ordered $\sigma$-structure $\mA$ is a finite structure of vocabulary $\sigma\cup\{\leq\}$, where $\leq\notin\sigma$ is a binary relation symbol and $\leq^\mA$ is a linear order on $A$.
Every finite ordered structure has a corresponding isomorphic structure, whose domain is an
initial segment of the natural numbers. Thus, we assume, as usual, that $A = \{0, 1, \ldots, n-1\}$, where $n$ is the cardinality $|A|$ of $A$.  

In this paper, $\log n$ always refers to the binary logarithm of $n$, i.e., $\log_2 n$. We write $\log^k n$ as a shorthand for $(\left\lceil\log n \right\rceil)^k$. A tuple of elements $(a_1,\dots,a_k)$ is sometimes written as $\bar{a}$. We then use $\bar{a}[i]$ to denote the $i$-th element of the tuple. Similarly, if $s$ is a finite string, we denote by $s[i]$ the $i$-th letter of this string. 

	\section{Deterministic polylogarithmic time}\label{sec:random-access}

		The sequential access that Turing machines have to their tapes restrict sub-linear time computations to depend only on the first sub-linear bits of the input; there is now way to access an arbitrary bit of the input.
		Therefore, logarithmic time complexity classes are usually studied using models of computation that have random-access\footnote{The term \emph{random-access} refers to the manner how \emph{random-access memory} (RAM) is read and written. In contrast to sequential memory, the time it takes to read or write using RAM is almost independent of the physical location of the data in the memory. We want to emphasise that there is nothing \emph{random} in random-access.} to their input, i.e., that can access every input address directly. As this also applies to polylogarithmic time, we adopt a Turing machine model that has a \emph{random-access} read-only input, similar to the logarithmic-time Turing machine in~\cite{barrington:jcss1990}. 

Our concept of a \emph{random-access Turing machine} is that of a
multi-tape Turing machine which consists of: (1) a finite set of
states, (2) a read-only random access \emph{input-tape}, (3) a
sequential access \emph{address-tape}, and (4) one or more (but a
fixed number of) sequential access \emph{work-tapes}. All tapes
are divided into cells, each equipped with a \emph{tape head}
which scans the cells, and are ``semi-infinite'' in the sense
that they have no rightmost cell, but have a leftmost cell. The
tape heads of the sequential access address-tape and
work-tapes can move left or right. When a head is in the
leftmost cell, it is not allowed to move left. The address-tape
alphabet only contains symbols $0$, $1$ and $\sqcup$ (for blank).
The position of the input-tape head is determined by the number
$i$ stored in binary between the leftmost cell and the first
blank cell of the address-tape (if the leftmost cell is blank, then $i$ is considered to be $0$) as follows: If $i$ is strictly
smaller than the length $n$ of the input string, then the
input-tape head is in the $(i+1)$-th cell. Otherwise, if $i \geq
n$, then the input-tape head is in the $(n+1)$-th cell scanning
the special end-marker symbol $\triangleleft$.

Formally, a \emph{random-access Turing machine} $M$ with $k$
work-tapes is a five-tuple $(Q, \Sigma, \delta, q_0, F)$.
Here $Q$ is a finite set of \emph{states}; $q_0 \in Q$ is the
\emph{initial state}. $\Sigma$ is a finite set of symbols (the
\emph{alphabet} of $M$). For simplicity, we fix $\Sigma = \{0, 1,
\sqcup\}$. $F \subseteq Q$  is the set of \emph{accepting final
states}. The \emph{transition} function of $M$ is of the form
$\delta : Q \times (\Sigma \cup
\{\triangleleft\}) \times \Sigma^{k+1} \rightarrow Q \times
(\Sigma \times \{\leftarrow, \rightarrow, - \})^{k+1}$.  
We assume that the tape head
directions $\leftarrow$ for ``left'', $\rightarrow$ for ``right''
and $-$ for ``stay'', are not in $Q \cup \Sigma$.

Intuitively, $\delta(q, a_1, a_2, \ldots, a_{k+2}) = (p, b_2,
D_2, \ldots, b_{k+2}, D_{k+2})$ means that, if $M$ is in the state
$q$, the input-tape head is scanning $a_1$, the index-tape head
is scanning $a_2$, and for every $i = 1, \ldots, k$ the head of
the $i$-th work-tape is scanning $a_{i+2}$, then the next state
will be $p$, the index-tape head will write $b_2$ and move in the
direction indicated by $D_2$, and for every $i = 1, \ldots, k$
the head of the $i$-th work-tape will write $b_{i+2}$ and move in
the direction indicated by $D_{i+2}$. Situations in which the
transition function is undefined indicate that the computation
must stop. Observe that $\delta$ cannot change the contents of
the input tape.

A \emph{configuration} of $M$ on a fixed input
$w_0$ is a $k+2$ tuple $(q, i, w_1, \ldots, w_{k})$, where $q$ is
the current state of $M$, $i \in \Sigma^*\# \Sigma^*$ represents
the current contents of the index-tape cells, and each $w_j \in
\Sigma^*\# \Sigma^*$ represents the current contents of the
$j$-th work-tape cells. We do not include the contents of the
input-tape cells in the configuration since they cannot be
changed. Further, the position of the input-tape head is uniquely
determined by the contents of the index-tape cells. The symbol
$\#$ (which we assume is not in $\Sigma$)  marks the position of
the corresponding tape head. By convention, the head scans the symbol
immediately at the right of $\#$. All symbols in the infinite
tapes not appearing in their corresponding strings $i, w_0,
\ldots, w_k$ are assumed to be the designated symbol for blank $\sqcup$.

At the beginning of a computation all work-tapes are blank,
except the input-tape, that contains the input string, and the
index-tape that contains a $0$ (meaning that the input-tape head
scans the first cell of the input-tape).  
Thus, the \emph{initial configuration} of $M$ is $(q_0, \#0, \#,
\ldots, \#)$. A \emph{computation} is a (possibly infinite) sequence of configurations which starts with the initial
configuration and, for every two consecutive configurations, the latter is obtained by applying the transition function of $M$ to the former.
An input string is \emph{accepted} if an accepting configuration, i.e., a configuration in which the current state belongs to $F$, is reached.

\begin{example}\label{example1}
Following a simple strategy, a random-access Turing machine $M$
can figure out the length $n$ of its input as well as $\lceil
\log n \rceil$ in polylogarithmic time. In its initial step, $M$
checks whether the input-tape head scans the end-marker
$\triangleleft$. If it does, then the input string is the empty
string and its work is done.  Otherwise, $M$ writes $1$ in the
first cell of its address tape and keeps writing $0$'s in its
subsequent cells right up until the input-tape head scans
$\triangleleft$. It then rewrites the last $0$ back to the blank symbol  $\sqcup$. At this point the resulting binary string in the
index-tape is of length $\lceil \log n \rceil$. Next, $M$ moves its
address-tape head back to the first cell (i.e., to the only cell containing a $1$ at this point). From here on, $M$ repeatedly moves the index head one step to the
right.  Each time it checks whether the index-tape head scans a
blank $\sqcup$ or a $0$. If $\sqcup$ then $M$ is done. If $0$, it
writes a $1$ and tests whether the input-tape head jumps to the
cell with $\triangleleft$; if so,  it rewrites a $0$, otherwise,
it leaves the $1$.  The binary number left on the index-tape at
the end of this process is $n-1$. Adding one in binary is now an
easy task. \qed
\end{example}

The \emph{formal language accepted} by a
machine $M$, denoted $L(M)$, is the set of strings accepted by
$M$. We say that $L(M) \in \mathrm{DTIME}[f(n)]$ if $M$ makes at
most $O(f(n))$ steps before accepting or rejecting an input
string of length $n$. 
We define the class of all formal languages decidable by (deterministic)
random-access Turing machines in \emph{polylogarithmic time} as
follows: \begin{equation*} \dpolylog = \bigcup_{k \in \mathbb{N}}
\mathrm{DTIME}[\log^k n] \qquad \end{equation*}
	
It follows from Example~\ref{example1} that a $\dpolylog$ random-access Turing machine can check any numerical
property that is polynomial time in the size of its input in
binary. For instance, it can check whether the length of its
input is even, by simply looking at its least-significant bit.

When we want to give a finite structure as an input to a
random-access Turing machine, we encode it as a string, adhering
to the usual conventions in descriptive complexity
theory~\cite{immerman_book}. Let $\sigma = \{R^{r_1}_1, \ldots,
R^{r_p}_p, c_1, \ldots, c_q, f^{k_1}_1, \ldots, f^{k_s}_s\}$ be a
vocabulary, and let ${\bf A}$ with $A = \{0, 1,{\dots},
{n{-}1}\}$ be an ordered structure of vocabulary $\sigma$. Note that the order on $A$ can be used to define an order for tuples of elements of $A$ as well. Each
relation $R_i^{\bf A} \subseteq A^{r_i}$ of $\bf A$ is encoded as
a binary string $\mathrm{bin}(R^{\bf A}_i)$ of length $n^{r_i}$,
where $1$ in a given position $m$ indicates that the $m$-th
tuple of $A^{r_i}$ is in $R_i^{\textbf{A}}$.  Likewise, each constant number
$c^{\bf A}_j$ is encoded as a binary string $\mathrm{bin}(c^{\bf
A}_j)$ of length $\lceil \log n \rceil$.

We also need to encode the functions of a structure.  We view $k$-ary
functions as consisting of $\lceil \log n \rceil$ many $k$-ary
relations, where the $m$-th relation indicates whether the $m$-th
bit of the value of the function is $1$. Thus, each function $f^{\bf A}_i$ is encoded as a
binary string $\mathrm{bin}(f^{\bf A}_i)$ of length $\lceil \log
n \rceil n^{k_i}$. 

The encoding of the whole structure $\mathrm{bin}(\textbf{A})$ is
the concatenation of the binary strings encoding its relations,
constants, and functions. The length $\hat{n} =
|\mathrm{bin}(\textbf{A})|$ of this string is
$n^{r_1}+\cdots+n^{r_p} + q \lceil \log n \rceil + \lceil \log n
\rceil n^{k_1}+\cdots+\lceil \log n \rceil n^{k_s}$, where $n =
|A|$ denotes the size of the input structure ${\bf A}$. Note that
$\log \hat{n} \in O(\lceil \log n \rceil)$, and hence
$\mathrm{DTIME}[\log^k \hat{n}] = \mathrm{DTIME}[\log^k n]$.

\section{Direct-access Turing machines}\label{datm}
	
In this section, we propose a new model of random-access Turing
machines. In the standard model reviewed above,
the entire input structure is
assumed to be encoded as one binary string.
In our new variant, the different relations and functions
of the structure can be accessed directly. We then show that both
variants are equivalent, in the sense that they lead to the same
notion of $\dpolylog$.  The direct-access model will then be useful to
give a transparent proof of our capturing result.

Let $\sigma = \{R^{r_1}_1, \ldots, R^{r_p}_p, c_1,
\ldots c_q, f^{k_1}_1, \ldots, f^{k_s}_s\}$ be a vocabulary. A
\emph{direct-access Turing machine that takes $\sigma$-structures
$\mA$ as an input}, 
is a multitape Turing machine with
$r_1 + \cdots + r_p + k_1 + \dots + k_s$ distinguished
work-tapes, called \emph{address-tapes}, $s$ distinguished
read-only (function) \emph{value-tapes}, $q+1$ distinguished
read-only \emph{constant-tapes}, and one or more ordinary
\emph{work-tapes}.

Let us define
  a transition function $\delta_l$
  for each tape $l$ separately.  These transition functions take
  as an input the current state of the machine, the bit read by
  each of the heads of the machine, and, for each relation
  $R_i\in \sigma$, the answer (0 or 1) to the query
$(n_1, \dots, n_{r_i}) \in R^\mA_i$.  Here,
$n_j$ denotes the number written in binary in
the $j$th distinguished tape of $R_i$.

Thus, with $m$ the total number of tapes,
the state transition function has the form
\[ \delta_Q: Q \times \Sigma^m\times \{0,1\}^p \rightarrow Q. \]
If $l$ corresponds to an
address-tape or an ordinary work-tape, we get the form
\[ \delta_l: Q \times
\Sigma^m\times \{0,1\}^p  \rightarrow \Sigma \times \{\leftarrow,
\rightarrow, - \}.  \] If $l$ corresponds to one of the read-only
tapes, we have \[ \delta_l: Q \times \Sigma^m\times \{0,1\}^p  \rightarrow
\{\leftarrow, \rightarrow, - \}.  \]

Finally we update the contents of the function value-tapes. If
$l$ is the function value-tape for a function $f_i$, then the
content of the tape $l$ is updated to $f^{\mA}_i(n_1,\dots n_{k_i})$
written in binary.  Here, $n_j$ denotes the number written in
binary in the $j$th distinguished address-tape of $f_i$
\emph{after} the execution of the above transition functions.  If
one of the $n_j$ is too large, the tape $l$ is updated to contain
only blanks. Note that the head of the tape remains in place; it
was moved by $\delta_l$ already.



In the initial configuration, read-only constant-tapes for the
constant symbols $c_1, \ldots, c_q$ hold their values in ${\bf A}$ in binary.  One additional
constant-tape (there are $q+1$ of them) holds the size $n$ of the domain of
${\bf A}$ in binary. Each address-tape, each
value-tape, and each ordinary work-tape holds only
blanks.

\begin{theorem}	\label{directrandom}
A class of finite ordered
structures $\cal C$ of some fixed vocabulary $\sigma$ is decidable by a random-access Turing machine working in
$\polylog$ with respect to $\hat{n}$,
where $\hat{n}$ is the size of the binary encoding of the input structure, iff $\cal C$ is decidable by a direct-access Turing machine in
$\polylog$ with respect to $n$,
where $n$ is the size of the domain of the input structure.
\end{theorem}

\begin{proof}
We will first sketch how a random-access Turing machine $M_r$ simulates a direct-access Turing machine $M_d$ on an input $\mA$. Let $n$ denote the cardinality of $A$ and $\en$ the length of $\enc(\mA)$. We dedicate a work-tape of $M_r$ to every tape of $M_d$. In addition, for each relation $R$, we add one extra tape that will always contain the answer to the query $?R(\vec{n})$. We also use additional work-tapes for convenience.  We then encode the initial configuration of $M_d$ into the tapes of $M_r$:
\begin{enumerate}
\item On the 0th constant tape,
write $n$ in binary.
\item On each tape for a constant $c_i$, write $c_{i}^{\mA}$ in binary.
\item For the answer-tapes of relations $R_i$, write the bit $0$.
\end{enumerate}
For encoding the transitions of $M_d$, we will in addition need two more constructs:
\begin{enumerate}[label=\alph*.]
\item Updating the answer-tapes of relations after each transition.
\item Updating the answer-tapes of functions after each transition.
\end{enumerate}
We now need to verify that these procedures (3.\ is trivial)
can be performed by $M_r$ in polylogarithmic time with respect to $\en$.

Step 1. On a fixed vocabulary $\sigma$, we have
$\en = f(n)$, for some fixed function $f$ of the form
\[
n^{r_1}+\cdots+n^{r_p} + q \lceil \log n \rceil + \lceil \log n \rceil n^{k_1}+\cdots+\lceil \log n \rceil n^{k_s}.
\]
We will find $n$ by executing a binary search between the numbers
$0$ and $\en$; note that checking whether a binary representation
of a number is at most $\en$, can be checked by writing the
representation to the index-tape and checking whether a bit or
$\triangleleft$ is read from the input-tape. For each i between
$0$ and $\en$,  $f(i)$ can be computed in polynomial time with respect to the length of $\en$ in binary, and thus in polylogarithmic time
with respect to $\en$.

Step 2.  The binary representation of a constant $c^\mA_i$ is written
in the input-tape between $g(n)$ and $g(n)+ \lceil \log n
\rceil$, where $g$ is a fixed function of the form \(
n^{r_1}+\cdots+n^{r_p} + (i-1) \lceil \log n \rceil.  \) The
numbers $n$ and $g(n)$ are obtained as in case 1. Then $g(n)$ is
written on the index tape and the next $\lceil \log n \rceil$
bits of the input are copied to the tape corresponding to
$c_i$.

Steps a.\ and b.\, These cases are are handled similar to each other and to the case 2. above.  The main
difference for b.\ is that the bits of the output are not in
successive positions of the input, but the location of each bit
needs to be calculated separately. 

We next sketch how a direct-access Turing machine $M_d$
simulates a random-access Turing machine $M_r$ on an input $\mA$.
First note that approach similar to the converse direction does
not work here, as we do not have enough time to directly
construct the initial configuration of $M_r$ inside  $M_d$. For
each work-tape of $M_r$, we dedicate a work-tape of $M_d$. For the
index-tape of $M_r$, we dedicate a work-tape of $M_d$ and call it
the index-tape of $M_d$. Moreover, we use some additional
work-tapes for convenience. The idea of the simulation is that
the dedicated work-tapes and the index-tape of $M_d$ copy exactly
the behaviour of the corresponding tapes of $M_r$. The additional
work-tapes are used to calculate to which part of the input of
$M_r$ the index-tape refers to. After each transition of $M_r$
this is checked so that the machine $M_d$ can update its
address-tapes accordingly.

Recall that given an input $\sigma = \{R^{r_1}_1, \ldots, R^{r_p}_p, c_1, \ldots c_q, f^{k_1}_1, \ldots, f^{k_s}_s\}$ structure $\mA$ of cardinality $n$, the input of $M_r$ is of length
\begin{equation}\label{eq:one}
 n^{r_1}+\cdots+n^{r_p} + q \lceil \log n \rceil + \lceil \log n \rceil n^{k_1}+\cdots+\lceil \log n \rceil n^{k_s}.
\end{equation}
The number written in binary on the index-tape of $M_r$ determines the position of the input that is read by $M_r$. From \eqref{eq:one} we obtain fixed functions on $n$, that we use in the simulation to check which part of the input is read when the index-tape holds a particular number.
For example,
if the index-tape holds  $n^r_1+1$, we can calculate that the
head of the input-tape of $M_r$ reads the bit answering the query:
is $\vec 0 \in R_2^{\mA}$.
We can use an extra work-tape of $M_d$ to always store the bit
that  $M_r$ is reading from its input; the rest of the simulation
is straightforward.
\end{proof}

\section{Index logic}{\label{sec:ifpplog}}

In this section, we introduce \emph{index logic}, a new logic
which over ordered finite structures captures $\dpolylog$. Our
definition of index logic is inspired by the second-order
logic in~\cite{Barr92}, where relation variables are restricted
to valuations on the sub-domain $\{0, \ldots, \lceil \log n
\rceil-1\}$ ($n$ being the size of the interpreting structure), as well as by the well known counting logics as defined in \cite{grohe_2017}.

Given a vocabulary $\sigma$, for every ordered $\sigma$-structure
$\mathbf{A}$, we define a corresponding set of natural numbers
$\textit{Num}(\mathbf{A}) = \{0,\dots,\lceil \log n \rceil-1\}$
where $n=|A|$. Note that $\textit{Num}(\mathbf{A}) \subseteq A$,
since we assume that $A$ is an initial segment of the natural
numbers. This simplifies the definitions, but it is otherwise unnecessary. 

Index logic is a two-sorted logic. Individual variables of the first sort \textbf{v} range over the domain $A$ of $\mathbf{A}$, while individual variables of the second sort \textbf{n} range over $\textit{Num}(\mathbf{A})$. We denote variables of sort \textbf{v} with  $x, y, z, \ldots$, possibly with a subindex   such as $x_0,x_1, x_2, \dots$, and variables of sort \textbf{n} with $\mathtt{x}, \mathtt{y}, \mathtt{z}$, also possibly with a subindex. Relation variables, denoted with uppercase letters $X,Y,Z, \ldots$, are always of sort \textbf{n}, and thus range over relations defined on $\textit{Num}(\mathbf{A})$.

\begin{definition}[Numerical and first-order terms]
The only terms of sort \emph{\textbf{n}} are the variables of sort  \emph{\textbf{n}}. For a vocabulary $\sigma$, the $\sigma$-terms $t$ of sort \emph{\textbf{v}} are generated by the following grammar:
\[
t ::= x \mid c \mid f(t, \ldots, t),
\]
where $x$ is a variable of sort \emph{\textbf{v}}, $c$ is a constant symbol in $\sigma$, and $f$ is a function symbol in $\sigma$.
\end{definition}

\begin{definition}[Syntax of index logic]\label{syntax}
Let $\sigma$ be a vocabulary. The formulae of \emph{index logic} $\mathrm{IL(IFP)}$ is generated by the following grammar:
\begin{multline*}
\varphi ::= t_1 \leq t_2 \mid \ttx_1 \leq \ttx_2 \mid R(t_1, \ldots, t_k) \mid X(\ttx_1, \ldots, \ttx_k) \mid  (\varphi \land \varphi) \mid \neg \varphi  \mid  [\mathrm{IFP}_{\bar{\mathtt{x}}, X} \varphi]\bar{\tty} \mid \\
t=\mathit{index}\{\mathtt{x}:\varphi(\mathtt{x})\} \mid \exists x (x=\mathit{index}\{\mathtt{x}:\alpha(\mathtt{x})\} \wedge \varphi)  \mid \exists \ttx \varphi,
\end{multline*}
where $t, t_1, \ldots, t_k$ are $\sigma$-terms of sort \emph{\textbf{v}}, $\ttx, \ttx_1, \ldots, \ttx_k$ are variables of sort \emph{\textbf{n}}, $\bar{\mathtt{x}}$ and $\bar{\tty}$ are tuples of variables of sort  \emph{\textbf{n}} whose length coincides with the arity of the relation variable $X$. Moreover, $\alpha(\ttx)$ is a formula where the variable $x$ of sort \emph{\textbf{v}} does not occur as a free variable.
\end{definition}
We also use the standard shorthand formulae  $t_1 = t_2$, $\ttx_1 = \ttx_2$, $(\varphi\lor \psi)$, and $\forall \tty \varphi$ with the obvious meanings.

The concept of a valuation is the standard one for a two-sorted logic. Thus, a \emph{valuation} over a structure $\mathbf{A}$ is any total function \textit{val} from the set of all variables of index logic to values satisfying the following constraints:
        \begin{itemize}
            \item If $x$ is a variable of sort \textbf{v}, then $\mathit{val}(x) \in A$.
            \item If $\mathtt{x}$ is a variable of sort \textbf{n}, then $\mathit{val}(\mathtt{x}) \in \textit{Num}(\mathbf{A})$.
            \item If $X$ is a relation variable with arity $r$, then $\mathit{val}(X) \subseteq (\textit{Num}(\mathbf{A}))^r$.
        \end{itemize}
        
If $\chi$ is a variable and $B$ a legal value for that variable, we write $\it{val}(B/\chi)$ to denote the valuation that maps $\chi$ to $B$ and agrees with $\it{val}$ for all other variables. Valuations extend to terms and tuples of terms in the usual way.
      
      Fixed points are defined in the standard way (see \cite{ef_fmt2} and~\cite{libkin_fmt} among others). Given an operator $F : {\cal P}(B) \rightarrow {\cal P}(B)$, a set $S \subseteq B$ is a \emph{fixed point} of $F$ if $F(S) = S$. A set $S \subseteq B$ is the \emph{least fixed point} of $F$ if it is a fixed point and, for every other fixed point $S'$ of $F$, we have $S \subseteq S'$. We denote the least fixed point of $F$ as $\mathrm{lfp}(F)$. The \emph{inflationary fixed point} of $F$, denoted by $\mathrm{ifp}(F)$, is the union of all sets $S^i$ where $S^0 := \emptyset$ and $S^{i+1} := S^i \cup F(S^i)$.   

Let $\varphi(X, \bar{\mathtt{x}})$ be a formula of vocabulary $\sigma$, where $X$ is a relation variable of arity $k$ and $\mathtt{x}$ is a $k$-tuple of variables of sort $\textbf{n}$. Let $\bf A$ be a $\sigma$-structure and $\it{val}$ a variable valuation. The formula $\varphi(X, \bar{\mathtt{x}})$ gives rise to an operator $F^{\bf A, \it{val}}_{\varphi,\bar{\mathtt{x}}, X} : {\cal P}((\textit{Num}(\mathbf{A}))^k) \rightarrow {\cal P}((\textit{Num}(\mathbf{A}))^k)$ defined as follows:
\[
F^{\bf A, \it{val}}_{\varphi, \bar{\mathtt{x}}, X}(S) := \{ \bar{a}\in (\textit{Num}(\mathbf{A}))^k \mid \mathbf{A},\mathit{val}(S/X, \bar{a}/ \bar{\mathtt{x}}) \models \varphi (X,\bar{\mathtt{x}}).
\]


\begin{definition}\label{semanticsIndexLogic}
      The formulae of $\mathrm{IL(IFP)}$ are interpreted as follows:
        \begin{itemize}
			\item $\mathbf{A},\mathit{val} \models \ttx_1 \leq \ttx_2$ iff $\mathit{val}(\ttx_1)  \leq \mathit{val}(\ttx_2)$.             
            \item $\mathbf{A},\mathit{val} \models t_1 \leq t_2$ iff $\mathit{val}(t_1) \leq \mathit{val}(t_2)$. 
            \item $\mathbf{A},\mathit{val} \models R(t_1,\dots, t_k) $ iff $(\mathit{val}(t_1),\dots,\mathit{val}(t_k))\in R^\mathbf{A}$.
            \item $\mathbf{A},\mathit{val} \models X(\ttx_1,\dots,\ttx_k) $ iff $(\mathit{val}(\ttx_1),\dots,\mathit{val}(\ttx_k)) \in \mathit{val}(X) $.
            \item $\mathbf{A},\mathit{val} \models t=\mathit{index}\{\mathtt{x}:\varphi(\mathtt{x})\}$ iff           
             $\mathit{val}(t)$ in binary is $b_m b_{m-1} \cdots b_0$, where $m = {\lceil\log |A|\rceil}-1$ and  $b_j = 1$ iff $\mathbf{A},\mathit{val}(j/\mathtt{x}) \models \varphi(\mathtt{x})$.   
            %
            \item $\mathbf{A},\mathit{val} \models [\mathrm{IFP}_{\bar{\mathtt{x}}, X} \varphi]\bar{\tty}$ iff  $\mathit{val}(\bar{\tty}) \in \mathrm{ifp}(F^{\bf A, \it{val}}_{\varphi, \bar{\mathtt{x}}, X})$.  

            \item $\mathbf{A},\mathit{val} \models \neg \varphi $ iff $\mathbf{A},\mathit{val} \not\models \varphi$.
            \item $\mathbf{A},\mathit{val} \models  \varphi \wedge \psi $ iff $\mathbf{A},\mathit{val} \models \varphi$ and $\mathbf{A},\mathit{val} \models \psi$.
           \item $\mathbf{A},\mathit{val} \models \exists \mathtt{x} \,  \varphi$ iff $\mathbf{A}, \mathit{val}(i/\mathtt{x}) \models \varphi$, for some $i\in\textit{Num}(\mathbf{A})$.

            \item $\mathbf{A},\mathit{val} \models \exists x (x=\mathit{index}\{\mathtt{x}:\alpha(\mathtt{x})\} \wedge \varphi)$ iff there exists $i\in A$ such that $\mathbf{A},\mathit{val}(i/\mathtt{x}) \models x=\mathit{index}\{\mathtt{x}:\alpha(\mathtt{x})\}$ and $\mathbf{A},\mathit{val}(i/\mathtt{x}) \models \varphi$. 

        \end{itemize}    
\end{definition} 

It immediately follows from the famous result by Gurevich and Shelah regarding the equivalence between inflationary and least fixed points~\cite{gs_fixpoint}, that an equivalent index logic can be obtained if we (1)~replace $[\mathrm{IFP}_{\bar{\mathtt{x}}, X} \varphi]\bar{\tty}$ by $[\mathrm{LFP}_{\bar{\mathtt{x}}, X} \varphi]\bar{\tty}$ in the formation rule for the fixed point operator in Definition~\ref{syntax}, adding the restriction that every occurrence of $X$ in $\varphi$ is positive\footnote{This ensures that $F^{\bf A, \it{val}}_{\varphi, \bar{\mathtt{x}}, X}$ is a monotonic function and that the least fixed point $\mathrm{lfp}(F^{\bf A,\it{val}}_{\varphi, \bar{\mathtt{x}}, X})$ exists.}, and (2)~fix the interpretation $\mathbf{A},\mathit{val} \models [\mathrm{LFP}_{\bar{\mathtt{x}}, X} \varphi]\bar{y}$ iff  $\mathit{val}(\bar{y}) \in \mathrm{lfp}(F^{\bf A, \it{val}}_{\varphi, \bar{\mathtt{x}}, X})$. 

Moreover, the convenient tool of \emph{simultaneous fixed points}, which allows one to iterate several formulae at once, can also be used here, since it does not increase the expressive power of the logic. Following the
syntax and semantics proposed by Ebbinghaus and Flum~\cite{ef_fmt2}, a version of index logic with simultaneous inflationary fixed point operators can be obtained by replacing the clause corresponding to $\mathrm{IFP}$ in Definition~\ref{syntax} by the following:
\begin{itemize}
\item If $\bar{\tty}$ is tuple of variables of sort \textbf{n}, and for $m \geq 0$ and $0 \leq i \leq m$, we have that $\bar{\mathtt{x}}_i$ is also a tuple of variables of sort \textbf{n}, $X_i$ is a relation variable whose arity coincides with the length of $\bar{\mathtt{x}}_i$, the lengths of $\bar{\tty}$ and $\bar{\mathtt{x}}_0$ are the same, and $\varphi_i$ is a formula, then $[\textrm{S-IFP}_{\bar{\mathtt{x}}_0, X_0, \ldots, \bar{\mathtt{x}}_m, X_m} \varphi_0, \ldots, \varphi_m]\bar{\tty}$ is an atomic formula.  
\end{itemize}
The interpretation is that $\mathbf{A},\mathit{val} \models [\textrm{S-IFP}_{\bar{\mathtt{x}}_0, X_0, \ldots, \bar{\mathtt{x}}_m, X_m} \varphi_0, \ldots, \varphi_m]\bar{\tty}$ iff $\mathit{val}(\bar{\tty})$ belongs to the first (here $X_0$) component of the simultaneous inflationary fixed point.

Thus, we can use index logic with the operators $\textrm{IFP}$, $\textrm{LFP}$, $\textrm{S-IFP}$ or $\textrm{S-LFP}$ interchangeably. 
%
%
%

In the next two subsections, we give two worked-out examples that illustrate the power of index logic. After that, the exact characterization of its expressive power is presented in Subsection~\ref{charProofIndexLogic}.
    
\subsection{Finding the binary representation of a term}
\label{binrepex}

Let $t$ be a term of \thefirstsort.  In this example,
we construct an index logic formula that expresses the well-known bit predicate
$\bit(t,\ttx)$. The predicate $\bit(t,\ttx)$ states that the $(\it{val}(\ttx)+1)$-th bit of $\it{val}(t)$ in binary is set.  Subsequently, the sentence $t
= \idx \{\ttx : \bit(t,\ttx)\} $ is valid over the class of all finite
ordered structures.

Informally, for a fixed term $t$, our implementation of $\bit(t,\ttx)$ works by iterating through the bit positions
$\tty$ from the most significant to the least significant.  These bits
are accumulated in a relation variable $Z$.  For each
$\tty$ we set the corresponding bit, on the condition that the
resulting number does not exceed $t$.  The set bits are collected
in a relation variable $Y$.

In the formal description of $\bit(t,\ttx)$ below, we use the
following abbreviations.  We use $M$ to denote the most
significant bit position.  Thus, formally, $\ttz = M$ abbreviates
$\forall \ttz'\, \ttz'\leq\ttz$.  Furthermore, for a unary
relation variable $Z$, we use $\ttz = \min Z$ with the obvious
meaning.  We also use abbreviations such as $\ttz=\ttz'-1$ with
the obvious meaning.

Now $\bit(t,\ttx)$ is a simultaneous fixed point
$ [\textrm{S-IFP}_{\tty,Y,\ttz,Z} \, \varphi_Y,\varphi_Z](\ttx)$,
where
\begin{align*}
\varphi_Z & := (Z = \emptyset \land \ttz=M) \lor
  (Z \neq \emptyset \land \ttz = \min Z - 1), \\
\varphi_Y & :=
  Z \neq \emptyset \land \tty = \min Z \land \exists x(x=\idx\{\ttz
  : Y(\ttz) \lor \ttz=\tty\} \land t \geq x).
\end{align*}

\subsection{Binary search in an array of key values}

In order to develop insight in how index logic works, we develop
in detail an example showing how binary search in an array of key
values can be expressed in the logic.

\newcommand{\KV}{K}
We represent the data structure as an ordered structure $\mA$
over the vocabulary consisting of a unary function $\KV$, a
constant symbol $N$, a constant symbol $T$, and a binary relation
$\prec$.  The domain of $\mA$ is an initial segment of the
natural numbers.  The constant $l:=N^\mA$ indicates the length of
the array; the domain elements $0$, $1$, \dots, $l-1$ represent
the cells of the array.  The remaining domain elements represent
key values.  Each array cell holds a key value; the assignment of
key values to array cells is given by the function ${\KV}^\mA$.

The simplicity of the above abstraction gives rise to two
peculiarities, which, however, pose no problems.  First, the
array cells belong to the range of the function $\KV$.  Thus, array
cells are allowed to play a double role as key values.  Second,
the function $\KV$ is total, so it is also defined on the domain
elements that are not array cells.  We will simply ignore $\KV$ on
that part of the domain.

We still need to discuss about $\prec$ and $T$. We assume $\prec^\mA$
to be a total order, used to compare key values.  So $\prec^\mA$
can be different from the built-in order $<^\mA$. For the binary
search procedure to work, the array needs to be sorted, i.e.,
$\mA$ must satisfy $ \forall x \forall y \Big(x<y < N \to \big(\KV(x) \preceq
\KV(y)\big)\Big)$.  Finally, the constant $t:=T^\mA$ is the test value.
Specifically, we are going to exhibit an index logic formula that
expresses that $t$ is a key value stored in the array. In other
words, we want to express the condition $$ \exists x (x < N \land
\KV(x)=T). \eqno (\gamma) $$  Note that, we express here the
condition $(\gamma)$ by a first-order
formula that is not an index logic formula.  So, our aim is to show
that $(\gamma)$ is still expressible, over all sorted arrays,
by a formula of index logic.

We recall the procedure for binary search \cite{knuth_vol3} in
the following form, using integer variables $L$, $R$ and $I$:
\begin{tabbing}
\hspace*{2em}\=\hspace*{2em}\=\hspace*{2em}\=\hspace*{2em}\kill
$L := 0$ \\
$R := N-1$ \\
\bf while $L\neq R$ do \+\\
$I := \lfloor (L+R)/2 \rfloor$ \\
\bf if $\KV(I) \succ T$ then $R := I-1$ else $L := I$  \-\\
{\bf if $\KV(L)=T$ return} `found' {\bf else return} `not found'
\end{tabbing}

We are going to express the above procedure as a simultaneous
fixed point, using binary relation variables
$L$ and $R$, and a unary relation variable $Z$.  We collect the
iteration numbers in $Z$, thus counting until the logarithm of
the size of the structure.  Relation variables $L$ and $R$ are
used to store the values, in binary representation, of the
integer variables $L$ and $R$ during all iterations.
Specifically, for each $i\in\Num(\mA)$, the value of the term
$\idx\{\ttx : L(i,\ttx)\}$ will be the value of the integer
variable $L$ before the $(i+1)$-th iteration of the while loop (and
similarly for $R$).

In the formal expression of $(\gamma)$ below, we use
the bit predicate from Section~\ref{binrepex}.
We also assume the following formulas:
\begin{itemize}
\item
A formula ${\it avg}(X,Y, \ttx)$ that expresses,
for unary relation variables $X$ and $Y$,
and a numeric variable $\ttx$,
that the bit $\ttx$ is set
in the binary representation of $\lfloor (x + y)/2 \rfloor$,
where $x$ and $y$ are the numbers represented in binary by $X$
and $Y$.
\item
A formula ${\it minusone}(X,\tty)$, expressing that the bit $\tty$ is
set in the binary representation of $x-1$, where $x$ is the
number represented in binary by $X$.
\end{itemize}
These formulas surely exist because index logic includes full
inflationary fixed point logic on the numeric sort; inflationary fixed point logic captures
PTIME on the numeric sort, and computing the average, or
subtracting one, are PTIME operations on binary numbers.

We are going to apply the formula $\it avg(X,Y,\ttx)$, where $X$ and $Y$ are
given by $L(\ttz,.)$ and $R(\ttz,.)$.
So, formally, below, we use $\it avg'(\ttz,\ttx)$ for the formula obtained
from the formula $\it avg$ by replacing each subformula of the form
$X(\ttu)$ by $L(\ttz,\ttu)$, and $Y(\ttu)$ by $R(\ttz,\ttu)$.

Furthermore, we are going to apply the formula $\it minusone(X, \ttu)$, where
$X$ is given by $\it avg'(\ttz)$.  So, formally, $\it minusone'(\ttz,\ttu)$ will
denote the formula obtained from $\it minusone(X,u)$ by replacing each
subformula of the form $X(\ttu)$ by ${\it avg}'(\ttz,\ttu)$.

A last abbreviation we will use is $\it test(\ttz)$, which will denote
the formula $\exists e (e = \idx\{\ttx : {\it avg}'(\ttz,\ttx)\} \land
\KV(e) \succ T)$.

Now $(\gamma)$ is expressed by $ \exists x (x=\idx\{\ttl :
\psi(\ttl)\} \land \KV(x)=T)$, where 
\begin{align*}
\psi(\ttl) & := \exists \tts \forall \tts' (\tts' \leq \tts \wedge [\textrm{S-IFP}_{\ttz,\ttx,L,\ttz,\ttx,R,\ttz,Z} \, \varphi_L,\varphi_R,\varphi_Z](\tts, \ttl)), \\
\varphi_Z & := (Z=\emptyset \land \ttz=0) \lor (Z\neq \emptyset
\land \ttz=\max Z + 1), \\
\varphi_L & := \begin{aligned}[t]
& Z \neq \emptyset \land \ttz=\max Z + 1 \land {}\\
& \exists \ttz' (\ttz' = \max Z \land
 ({\it test}(\ttz') \to L(\ttz',\ttx)) \land (\neg {\it test}(\ttz') \to {\it avg}'(\ttz', \ttx))),
\end{aligned}\\
\varphi_R & := \begin{aligned}[t]
& (Z=\emptyset \land \ttz=0 \land \bit(N-1,\ttx)) \lor (Z \neq \emptyset \land \ttz=\max Z + 1 \wedge {} \\
& \exists \ttz' (\ttz' = \max Z \land 
 ({\it test}(\ttz') \to {\it minusone}'(\ttz',\ttx)) \land (\neg {\it test}(\ttz') \to
R(\ttz',\ttx)))).
\end{aligned}
\end{align*}

\subsection{The logical characterization theorem for $\dpolylog$}\label{charProofIndexLogic}

The following result confirms that our logic serves our original purpose.

    \begin{theorem}\label{captureResult}    
    Over ordered structures, index logic captures $\dpolylog$.
    \end{theorem}

\begin{proof}${}$
%
%
\paragraph{Formulas of index logic can be evaluated in polylogarithmic time}
Let $\var$ be a finite set of variables (of sort {\bf{n}}, {\bf v}, and relational). We stipulate a Turing machine model that has a designated work-tape for each of the variables in $\var$. The idea here is that the tape designated for a variable contains the value of that variable encoded as a binary string.
We use induction on the structure of formulas to show that, for every sentence $\varphi$ of index logic, whose variables are from the set $\var$,  there exists a direct-access Turing machine $M_\varphi$ that, for every ordered structure  $\bf A$ with $|A| = n$, and every valuation $\mathit{val}$, 
decides in time $O(\lceil \log n \rceil^{O(1)})$ whether ${\bf A}, \mathit{val} \models \varphi$. Since $\var$ is an arbitrary finite set, this suffices.

%
In the proof, variables $v$ of sort \emph{\textbf{n}} and \emph{\textbf{v}} are treated in a similar way as constant symbols, meaning that their value $\mathit{val}(v)$ is written in binary in the first $\lceil \log n \rceil$ cells of their designated work-tapes. 
The work-tape designated to a relation variable $X$ of arity $k$ contains $\mathit{val}(X) \subseteq \mathit{Num}({\bf A})^{k}$ encoded as a binary string in its first $\lceil \log n \rceil^k$ cells, where a $1$ in the $i$-th cell indicates that the $i$-th tuple in the lexicographic order of $\mathit{Num}({\bf A})^{k}$ is in $\mathit{val}(X)$.

We will show first, by induction on the structure of terms, that, if $t$ is term, $M$ a direct-access Turing machine, and $\mathit{val}$ a valuation such that, for every variable $\chi$ that occurs in $t$, the value $\mathit{val}(\chi)$ is written in binary in the designated work-tape of $\chi$,  then $\mathit{val}(t)$ can be computed by $M$ in time $O(\lceil \log n \rceil^{O(1)})$.
%
 If $t$ is a variable of sort \emph{\textbf{n}} or \emph{\textbf{v}}, or a constant symbol, then $M$ only needs to read the first $\lceil \log n \rceil$ cells of the appropriate work-tape or constant-tape, respectively. If $t$ is a term of the form $f_i(t_1, \ldots, t_k)$, we access and copy each $\mathit{val}(t_j)$ in binary in the corresponding address-tapes of $f_i$. By the induction hypothesis, this takes time $O(\lceil \log n \rceil^{O(1)})$ each. Using $\lceil \log n \rceil$ additional steps the result of length $\lceil \log n \rceil$ will then be accessible in the value-tape of $f_i$.

We will next use induction to prove our main claim. Note that, the cases for quantifiers assure that the assumptions needed for the calculation of the values of terms are met. We will show by induction that, if $\varphi$ is a formula with variables in $\var$,  ${\it{val}}$ a valuation, and $M$ a direct-access Turing machine, such that, for every variable $\chi$ that occurs free in $\varphi$, the value $\mathit{val}(\chi)$ is written in binary in the designated work-tape of $\chi$,  then ${\bf A}, \mathit{val}\models \psi$ can be decided by $M$ in time $O(\lceil \log n \rceil^{O(1)})$.

%

If $\varphi$ is an atomic formula of the form $t_1 \leq t_2$, $M$ can evaluate $\varphi$ in polylogarithmic time by accessing the values of $t_1$ and $t_2$ in binary and then comparing their $\lceil \log n \rceil$ bits. 
    
    If $\varphi$ is an atomic formula of the form $R_i(t_1,\dots, t_k)$, $M$ can evaluate $\varphi$ in polylogarithmic time by simply computing the values of the terms $t_1,\dots, t_k$ and copying the values to the corresponding address-tapes of $R_i$. By the proof for terms above, computing the values of the terms take polylogarithmic time each, and since the values have $\lceil \log n \rceil$ bits, also the copying can be done in polylogarithmic time.
    
    If $\varphi$ is an atomic formula of the form $X(\ttx_1,\dots,\ttx_k)$, $M$ can evaluate $\varphi$ in polylogarithmic time by accessing the values $\ttx_1,\dots, \ttx_k$ in binary, computing the position $i$ of the tuple $(\ttx_1,\dots, \ttx_k)$ in the lexicographic order of $\mathit{Num}({\bf A})^{k}$ in binary, and then accessing the $i$-th cell of the work-tape which contains the encoding of $\mathit{val}(X)$ of length $\lceil \log n \rceil^k$. Computing $i$ in binary involves simple arithmetic operations on binary numbers of length bounded by $\log (\lceil \log n \rceil^k)$, which can clearly be done in time polynomial in $\log n$. 
    
If $\varphi$ is an atomic formula of the form $t=\mathit{index}\{\mathtt{x}:\psi(\mathtt{x})\}$, $M$ proceeds as follows. Let $s = \lceil \log n \rceil - 1$ and let $b_s b_{s-1}\cdots b_0$ be $\mathit{val}(t)$ in binary. For every $i$, $0\leq i \leq s$,
$M$ writes $i$ in binary in the work-tape designated for the variable $\mathtt{x}$  and checks whether ${\bf A}, \mathit{val}(i/\mathtt{x}) \models \psi(\mathtt{x})$ iff $b_i = 1$. Since, by the induction hypothesis, this check can be done in polylogarithmic time, and $\mathit{val}(t)$ can be computed in polylogarithmic time,  we get that $M$ decides $t=\mathit{index}\{\mathtt{x}:\varphi(\mathtt{x})\}$ in polylogarithmic time as well. 

If $\varphi$ is a formula of the form $[\mathrm{IFP}_{\bar{\mathtt{x}}, X} \psi]\bar{y}$, where the arity of $X$ is $k$,
let $F^{\bf A, \it{val}}_{\psi,\bar{\mathtt{x}}, X} : {\cal P}((\textit{Num}(\mathbf{A}))^k) \rightarrow {\cal P}((\textit{Num}(\mathbf{A}))^k)$ denote the related operator, $F^0 :=  \emptyset$, and $F^{i+1} := F^{i} \cup F^{\bf A, \it{val}}_{\psi,\bar{\mathtt{x}}, X} (F^i)$, for each $i\geq 0$. The inflationary fixed point is reached on stage $\lvert \textit{Num}(\mathbf{A})^k \rvert$, at the latest, and thus
$
\mathrm{ifp}(F^{\bf A, \it{val}}_{\psi,\bar{\mathtt{x}}, X}) = F^{\log^k n}.
$
Recall that
\[
F^{\bf A, \it{val}}_{\psi, \bar{\mathtt{x}}, X}(S) := \{ \bar{a}\in (\textit{Num}(\mathbf{A}))^k \mid \mathbf{A},\mathit{val}(S/X, \bar{a}/ \bar{\mathtt{x}}) \models \psi (X,\bar{\mathtt{x}})\}.
\]
We calculate $F^{i+1}$ from  $F^{i}$ as follows. Note that on each stage, the value of $F^{i}$ is written in binary on the work-tape designated for $X$. We first calculate the value of $F^{i+1}$ in binary on another work-tape, and then reformat the contents of the work-tape designated for $X$ to contain the value of $F^{i+1}$. For $i=0$, we format the work-tape designated for $X$ to contain a string of $0$s of length $\log^k n$. In order to calculate $F^{i+1}$ from $F^{i}$, we go through all $k$-tuples $\bar{a}\in (\textit{Num}(\mathbf{A}))^k$ in the lexicographic order. For $1 \leq j\leq k$, we write $\bar{a}[j]$ in binary on the designated work-tape for $\bar{\ttx}[j]$ and check whether
\begin{equation}\label{eq:1}
\mathbf{A},\mathit{val}(S/X, \bar{a}/ \bar{\mathtt{x}}) \models \psi (X,\bar{\mathtt{x}})
\end{equation}
holds. By induction hypothesis, this can be checked in time $O(\lceil \log n \rceil^{O(1)})$. If $\eqref{eq:1}$ holds and  $\bar{a}$ is the $l$-th k-tuple in the lexicographic ordering, we write $1$ to the $l$-th cell of the work-tape, where the value of $F^{i+1}$ is being constructed, otherwise we write $0$ to this cell. Hence the computation of $F^{i+1}$ from $F^{i}$ can be done in time $\log^k n \times O(\lceil \log n \rceil^{O(1)})$ which is still $O(\lceil \log n \rceil^{O(1)})$. It is now clear that $\mathrm{ifp}(F^{\bf A, \it{val}}_{\psi,\bar{\mathtt{x}}, X}) = F^{\log^k n}$ can be computed in time $O(\lceil \log n \rceil^{O(1)})$ as well. Finally, determining whether ${\it val}(\bar{y})$ is included in the fixed point is clearly computable in $O(\lceil \log n \rceil^{O(1)})$, for one must just calculate the position of ${\it val}(\bar{y})$ in the lexicographic order of $k$-tuples, and then check whether that position has a $0$ or $1$ in the work-tape corresponding to $X$.


If $\varphi$ is a formula of the form $\exists x (x=\mathit{index}\{\mathtt{x}:\alpha(\mathtt{x})\} \wedge \psi(x))$, $M$ proceeds as follows. For each $i \in \{0, \ldots, \lceil \log n \rceil-1\}$, $M$ writes $i$ in binary in the work-tape designated for $\mathtt{x}$ and checks whether ${\bf A}, \mathit{val}(i/\mathtt{x}) \models \alpha(\mathtt{x})$.
Since, by definition, $x$ does not appear free in $\alpha(\mathtt{x})$, it follows by the induction hypothesis that $M$ can perform each of these checks in polylogarithmic time. In parallel, $M$
writes the bit string $b_s b_{s-1}\cdots b_0$, defined such that $b_i = 1$ iff ${\bf A},\mathit{val}(i/\mathtt{x}) \models \alpha(\mathtt{x})$, to the work-tape designated to the variable $x$.
Let the content of this work-tape at the end of this process be $t$ in binary. $M$ can now check whether $t < n$ (recall that by convention, $M$ has the value $n$ in binary in one of its constant-tapes and thus this can be done in polylogarithmic time). If $t \geq n$ then ${\bf A}, \mathit{val} \not\models \varphi$.  If $t < n$, then $M$ checks whether ${\bf A}, \mathit{val}(t/x) \models \psi$, 
which by the induction hypothesis can also be done in polylogarithmic time. 

Finally, if $\varphi$ is a formula of the form $\exists \mathtt{x} \,\psi$, then for each $i \in \{0, \ldots, \lceil \log n \rceil-1\}$, $M$ writes $i$ in binary to the work-tape designated for $\mathtt{x}$ and checks whether ${\bf A}, \mathit{val}(i/\mathtt{x}) \models \psi$.
It follows by the induction hypothesis that $M$ can perform each of these checks in polylogarithmic time. If the test is positive for some $i$ then ${\bf A}, \mathit{val} \models \varphi$.
The remaining cases are those corresponding to Boolean connectives and follow trivially from the induction hypothesis.


\paragraph{Every polylogarithmic time property can be expressed in index logic}
Suppose we are given a class $\cal C$ of ordered $\sigma$-structures, which can be decided by a deterministic polylogarithmic time direct-access Turing machine $M = (Q, \Sigma, \delta, q_0, F, \sigma)$, that has $m$ tapes, including ordinary work-tapes, address-tapes, (function) value-tapes and constant-tapes. We assume, w.l.o.g., that $F = \{q_a\}$ (i.e., there is only one accepting state), $|Q| = a+1$, and $Q = \{q_0, q_1, \ldots, q_{a}\}$. 

Let $M$ run in time $O(\lceil \log n \rceil^k)$. Note that, only small inputs (up to some fixed constant) may require more time than $\lceil \log n \rceil^k$. Those finite number of small input structures can be dealt separately, for each finite structure can be easily defined by an index logic sentence. Hence, from now on, we only consider those inputs for which $M$ runs in time $\lceil \log n \rceil^k$.  Using the order relation $\leq^{\bf A}$ of the ordered structure $\bf{A}$, we can define the lexicographic order $\leq^\mA_k$ for the $k$-tuples in $\mathit{Num}(\mA)^k$, and then use this order to model time and positions of the tape heads of $M$. Note that this can be done, since the number of $k$-tuples in $\mathit{Num}(\mA)^k$ is $\lceil \log n \rceil^k$.  In our proof, we use expressions of the form $\bar{t} \sim t'$, where $\bar{t}$ is a $k$-tuple of variables of sort $\bf n$ and $t'$ is a single variable also of sort $\bf n$, with the intended meaning that $\it{val}(\bar{t})$ is the $(\it{val}(t')+1)$-th tuple in the order $\leq^\mA_k$. This is clearly expressible in index logic, since it is a polynomial time property on the $\bf n$ sort.

Next we introduce, together with their intended meanings, the relations we use to encode the configurations of  polylogarithmic time direct-access Turing machines.
Consider:
\begin{itemize}
\item A $k$-ary relation $S_q$, for every state $q \in Q$, such that $S_q(\bar{t})$ holds iff $M$ is in state $q$ at time $\bar{t}$. 
\item $2k$-ary relations $T_i^0, T_i^1, T_i^\sqcup$, for every tape $i = 1, \ldots, m$, such that $T_i^s(\bar{p}, \bar{t})$ holds iff at the time $\bar{t}$ the cell $\bar{p}$ of the tape $i$ contains the symbol $s$. 
\item $2k$-ary relations $H_i$, for every tape $i = 1, \ldots, m$, such that $H_i(\bar{p}, \bar{t})$ holds iff at the time $\bar{t}$ the head of the tape $i$ is on the cell $\bar{p}$.
\end{itemize}

We show that these relations are definable in index logic by means of a simultaneous inflationary fixed point formula. The following sentence is satisfied by a structure $\bf A$ iff ${\bf A} \in {\cal C}$. The idea of the formula is that it uses the simultaneous fixed point operator to construct the whole computation of $M$ iteration by iteration, and states that there exists a time step in which $M$ accepts. We define the formula
\begin{equation*}\label{IFPformula}
\exists \mathtt{x}_0 \ldots \mathtt{x}_{k-1} \big([\textrm{S-IFP}_{
\bar{t}, S_{q_a}, \mathrm{A}, \mathrm{B}_1, \mathrm{B}_2, \mathrm{B}_3, \mathrm{C} 
} 
\;\varphi_{q_a}, \Phi_\mathrm{A}, \Phi_{\mathrm{B}_1}, \Phi_{\mathrm{B}_2}, \Phi_{\mathrm{B}_3}, \Phi_{\mathrm{C}}](\mathtt{x}_0, \ldots, \mathtt{x}_{k-1})\big),
\end{equation*}   
where
\[\mathrm{A} = \bar{t}, S_{q_0}, \ldots, \bar{t}, S_{q_{a-1}} \quad
 \mathrm{B}_1 = \bar{p}\, \bar{t}, T^0_1, \ldots, \bar{p}\, \bar{t}, T^0_m \quad
\mathrm{B}_2 = \bar{p}\, \bar{t}, T^1_1, \ldots, \bar{p}\, \bar{t}, T^1_m
\]
\[\mathrm{B}_3 = \bar{p}\, \bar{t}, T^\sqcup_1, \ldots, \bar{p}\, \bar{t}, T^\sqcup_m \quad
\mathrm{C} = \bar{p}\, \bar{t}, H_1, \ldots, \bar{p}\, \bar{t}, H_m
\]
\[\Phi_\mathrm{A} = \varphi_{q_0}, \ldots, \varphi_{q_{a-1}} \quad
\Phi_{\mathrm{B}_1} = \psi_{01}, \ldots, \psi_{0m} \quad 
\Phi_{\mathrm{B}_2} = \psi_{11}, \ldots, \psi_{1m}\]
\[\Phi_{\mathrm{B}_3} = \psi_{\sqcup1}, \ldots, \psi_{\sqcup m}\quad
\Phi_{\mathrm{C}} = \gamma_{1}, \ldots, \gamma_{m}.\]
Note that here $\bar{p}$ and $\bar{t}$ denote $k$-tuples of variables of sort $\bf n$.

The formula builds the required relations $S_{q_i}$, $T^0_i$, $T^1_i$, $T^\sqcup_i$ and $H_i$ (for $1 \leq i \leq m$) in stages, where the $j$-th stage represents the configuration at time steps up to $j-1$. The subformulae $\varphi_{q_i}$, $\psi_{0i}$, $\psi_{1i}$, $\psi_{\sqcup i}$ and $\gamma_i$ define $S_{q_i}$, $T^0_i$, $T^1_i$, $T^\sqcup_i$ and $H_i$, respectively.

To simplify the presentation of the subformulae and w.l.o.g., we assume that, in every non-initial state of a computation, each address-tape contains a single binary number between $0$ and $n-1$ and nothing else. This number has at most $\log n$ bits, and hence we encode positions of address-tapes (and function value-tapes) with a single variable of sort $\bf n$ (instead of a tuple of variables).

We will now give the idea how the formulae  $\varphi_{q_i}$, $\psi_{0i}$, $\psi_{1i}$, $\psi_{\sqcup i}$, and $\gamma_i$ are constructed from $M$. We first describe the construction of $\psi_{0i}$ in detail; the formulae $\psi_{1i}$ and $\psi_{\sqcup i}$ are constructed in a similar fashion.
The rough idea behind all the formulas is the following: the formulas encode directly the initial configuration of the computation, and for a non-initial time step, how the configuration at that time step is computed from the previous configuration. The formula $\psi_{0i}(\bar{p},\bar{t})$, for example, encodes whether the $i$-th tape at the cell position $\bar{p}$ at the time $\bar{t}$ contains the symbol $0$. If $i$ is an address-tape or an ordinary work-tape, then in the initial configuration of the computation, the tape $i$ contains the blank symbol $\sqcup$ on all its cells.
In this case, the formula $\psi_{0i}$ is of the form:
\[
\neg(\bar{t} \sim 0) \wedge \alpha^0_i(\bar{p}, \bar{t} - 1),
\]
where $\alpha^0_i(\bar{p}, \bar{t} - 1)$ list conditions under which at the following time instant, $\bar{t}$, the position $\bar{p}$ of the tape $i$ will contain $0$. In the more general case, the formula has the form $(\bar{t}\sim 0 \land \xi_{T^0_i}) \lor (\neg(\bar{t} \sim 0) \wedge \alpha^0_i(\bar{p}, \bar{t} - 1))$, where $\xi_{T^0_i}$ is used to encode the initial configuration related to the relation $T^0_i$.

We will next describe the construction of $\alpha^0_i(\bar{p}, \bar{t} - 1)$.
Suppose, $i$ refers to an address-tape or to an ordinary work-tape.
The formula $\alpha^0_i(\bar{p}, \bar{t} - 1)$ is a disjunction over all the possible reasons, for why at the time $\bar{t}$ the position $\bar{p}$ of tape $i$ contains the symbol $0$. There are two possibilities: (1) at the time $\bar{t}-1$ the head of the tape $i$ was not in the position $\bar{p}$ and the position $\bar{p}$ of the tape $i$ contained the symbol $0$, (2)  at the time $\bar{t}-1$ the head of the tape $i$ was in the position $\bar{p}$ and the head wrote the symbol $0$. Below, we display a disjunct of $\alpha^0_i(\bar{p}, \bar{t} - 1)$ that is due to a reason of the second kind by one possible transition $\delta_i(q, a_1, \ldots, a_m, b_1, \ldots, b_p)=(0, \rightarrow)$. The disjunct of $\alpha^0_i(\bar{p}, \bar{t} - 1)$, which takes care of this case is obtained from the following formula by substituting $\bar{p}_i$ with $\bar{p}$:

{
\centering
\begin{minipage}{.65\textwidth}
\begin{align*}
&\exists \bar{p}_1 \dots \bar{p}_{i-1}\bar{p}_{i+1} \dots \bar{p}_m \Big(S_q(\bar{t}-1) \wedge\\
&\big(\bigwedge_{1\leq j\leq m} H_j(\bar{p}_j, \bar{t}-1) \wedge T^{a_j}_j(\bar{p}_j, \bar{t}-1) \big) \wedge \\
&  \bigwedge_{1\leq l\leq p}  \exists x_1 \ldots x_{r_l} \big( \mathrm{check}(R_l(x_1, \ldots, x_{r_l}), b_l) \land \\
&\quad \bigwedge_{1\leq k \leq r_l} x_k = \mathit{index}\{\mathtt{x} \mid (T^1_{\tau^R_{l,k}}(\mathtt{x}, \bar{t}-1))\} \big) \Big),
\end{align*}
\end{minipage}
\begin{minipage}{.325\textwidth}
\emph{At time $\bar{t}-1$, $M$ is in the state $q$ and the head of the tape $j$ is in the position $\bar{p}_j$ reading $a_j$.}

\vspace{1mm}
\emph{At time $\bar{t}-1$, the tuple of values in the address-tapes of $R_l$ is in $R^\mA$ iff $b_l=1$.}

\end{minipage}
}

\vspace{1mm}
\noindent where $\tau^R_{l,1}, \ldots, \tau^R_{l,r_l}$ denote the $r_l$ address-tapes corresponding to the $r_l$-ary relation $R_l$, and  $\mathrm{check}(R_l(x_1, \ldots, x_{r_l}), b_l)$ is a shorthand for $R_l(x_1, \ldots, x_{r_l})$, if $b_l=1$, and a shorthand for $\neg R_l(x_1, \ldots, x_{r_l})$, if $b_l=0$.

Assume then that $i$ refers to a value-tape of a function $f_j$ of arity $k_j$, and let $\tau^f_{j,1}, \ldots, \tau^f_{j,k_j}$ refer to its address-tapes.
Recall that the contents of a value-tape of a function at a time $\bar{t}$ depends only on the contents of its address-tapes at the time $\bar{t}$. Below, we write $\psi_{0i}(p,\bar{t})$ using the contents of the related address-tapes at time $\bar{t}$. This is fine, for we do not introduce circularity of definitions (technically, we obtain the contents of the related address-tapes at time $\bar{t}$ using the corresponding formulas that define them from the configuration of the machine at time $\bar{t}-1$).
Now $\psi_{0i}(p,\bar{t})$ refers to the following formula:
\begin{align*}
\exists x_1 \ldots x_{k_j} \Big( \big( \bigwedge_{1\leq l \leq k_j}  x_l = \mathit{index}\{\mathtt{x} \mid T^1_{\tau^f_{j,l}}(\ttx, \bar{t})\} \big) \land 
\neg \bit(f_j(x_1, \ldots, x_{k_j}),p) \Big),
\end{align*}
where $\bit(f_j(x_1, \ldots, x_{k_j}),p)$ expresses that the bit of position $p$ of $f_j(x_1, \ldots, x_{k_j})$ in binary is $1$; we showed, in Section~\ref{binrepex}, how the bit predicate is expressed in index logic.

The formula $\varphi_{q_0}$ is of the form $\bar{t} \sim 0 \vee (\neg(\bar{t} \sim 0) \wedge \alpha_{q_0}(\bar{t}-1))$ and other $\varphi_{q}$'s are of the form $\neg (\bar{t} \sim 0) \wedge \alpha_{q}(\bar{t}-1)$, where $\alpha_{q}(\bar{t}-1)$ list conditions under which $M$ will enter state $q$ at the next time instant, $\bar{t}$. 

Finally, the formulae $\gamma_i$ are of the form
\[
(\bar{t} \sim 0 \wedge \bar{p} \sim 0 ) \vee \big(\neg(\bar{t} \sim 0) \wedge \alpha_i(\bar{p}, \bar{t} - 1)\big),
\]
where $\alpha_i(\bar{p}, \bar{t} - 1)$ list conditions under which, at the following time instant $\bar{t}$, the head of the tape $i$ will be in the position $\bar{p}$. 

We omit writing the remaining subformulae, since it is an easy but tedious task. It is also not difficult to see that in the $j$-th stage of the simultaneous inflationary fixed point computation, the relations $S_q$, $(T_i^0, T_i^1, T_i^\sqcup)_{1 \leq i \leq m}$ and $(H_i)_{1 \leq i \leq m}$ encode the configuration of $M$ for times $\leq j-1$, which completes our proof.
\end{proof}

\section{Definability in Deterministic PolylogTime}

We observe here that very simple properties of structures
are nondefinable in index logic. Moreover, we provide an answer to a fundamental
question on the primitivity of the built-in order predicate (on
terms of \thefirstsort) in
our logic.   Indeed, we are working with ordered structures, and
variables of \thefirstsort{} can only be introduced by binding
them to an index term.  Index terms are based on sets of bit
positions which can be compared as binary numbers.  Hence, it is
plausible to suggest that the built-in order predicate can be removed 
from our logic without losing expressive power.  We prove,
however, that this does not work in the presence of constant or
function symbols in the vocabulary.

\begin{proposition} \label{emptiness}
Assume that the vocabulary includes a unary relation symbol $P$.
Checking emptiness (or non-emptiness) of $P^\mA$ in a given
structure $\mA$ is not computable in $\dpolylog$.
\end{proposition}
\begin{proof}
We will show that emptiness is not computable in $\dpolylog$. For a contradiction, assume that it is. Consider first-order structures over the vocabulary $\{P\}$, where $P$ is a unary relation symbol. Let $M$ be some Turing machine that decides in $\dpolylog$, given a $\{P\}$-structure $\mA$, whether $P^\mA$ is empty. Let $f$ be a polylogarithmic function that bounds the running time of $M$. Let $n$ be a natural number such that $f(n)< n$.

Let $\mA_\emptyset$ be the $\{P\}$-structure with domain $\{0,\dots,n-1\}$, where $P^\mA=\emptyset$. The encoding of $\mA_\emptyset$ to the Turing machine $M$ is the sequence
\(
 s := \underbrace{0\dots 0}_{\text{$n$ times}}
\).
Note that the running time of $M$ with input $s$ is strictly less than $n$. This means that there must exist an index $i$ of $s$ that was not read in the computation $M(s)$. Define 
\[
s' := \underbrace{0\dots 0}_{\text{$i$ times}} 1 \underbrace{0\dots 0}_{\text{$n-i-1$ times}}.
\]
Clearly the output of the computations $M(s)$ and $M(s')$ are identical, which is a contradiction since $s'$ is an encoding of a $\{P\}$-structure where the interpretation of  $P$ is a singleton.
\end{proof}
The technique of the above proof can be adapted to prove a plethora of undefinability results, e.g., it can be shown that $k$-regularity of directed graphs cannot be decided in $\dpolylog$, for any fixed $k$.

We can develop this technique further to show that the order
predicate on terms of \thefirstsort{} is a primitive in the logic.
The proof of the following lemma is quite a bit more complicated though. 

\begin{lemma}\label{lemma:order}
Let $P$ and $Q$ be unary relation symbols. There does not exist an index logic formula $\varphi$ such that for all $\{P,Q\}$-structures $\mA$ such that $P^\mA$ and $Q^\mA$ are disjoint singleton sets $\{l\}$ and $\{m\}$, respectively, it holds that
\[
\mA, \val \models \varphi \text{ if and only if }  l < m.
\]
\end{lemma}

\begin{proof}
We will show that the property described above cannot be decided
in $\polylog$; the claim then follows from
Theorem~\ref{captureResult}.
For a contradiction, suppose that the property can be decided in $\polylog$, and let $M$ and $f:\N\rightarrow \N$ be the related random-access Turing machine and polylogarithmic function, respectively, such that, for all $\{P,Q\}$-structures $\mA$ that satisfy the conditions of the claim,  $M(\enc(\mA))$ decides the property in at most $f(\vert \enc(\mA)\rvert)$ steps. Let $k$ be a natural number such that $f(2k) < k-1$.

Consider a computation $M(s)$ of $M$ with an input string $s$. We say that an index $i$ is \emph{inspected} in the computation, if at some point during the computation $i$ is written in the index tape in binary. Let $\ins_M(s)$ denote the set of inspected indices of the computation of $M(s)$ and $\ins^j_M(s)$ denote the set of inspected indices during the first $j$ steps of the computation.
%
%
We say that $s$ and $t$ are \emph{$M$-$j$-equivalent} if the lengths of $t$ and $s$ are equal and $t[i]=s[i]$, for each $i\in \ins^j_M(s)$. We say that $\mA$ and $\mB$ are \emph{$M$-$j$-equivalent} whenever $\enc(\mA)$ and $\enc(\mB)$ are. Note that if two structures $\mA$ and $\mB$ are $M$-$j$-equivalent, then the computations $M(\enc(\mA))$ and $M(\enc(\mB))$ are at the same configuration after $j$ steps of computation. Hence if $\mA$ and $\mB$ are M-$f(\vert \enc(\mA) \rvert)$-equivalent, then outputs of $M(\mA)$ and $M(\mB)$ are identical.

Let $\cC$ be the class of all $\{P,Q\}$-structures $\mA$ of domain $\{0,\dots k-1\}$, for which $P^\mA$ and $Q^\mA$ are disjoint singleton sets. The encodings of these structures are bit strings of the form $b_1\dots b_k c_1\dots c_k$, where exactly one $b_i$ and one $c_j$, $i\neq j$, is $1$.
The computation of $M(\enc(\mA))$ takes at most $f(2k)$ steps.

We will next construct a subclass $\cC^*$ of $\cC$ that consists of exactly those structures $\mA$ in $\cC$ for which the indices in $\ins(\enc(\mA))$ hold only the bit $0$. We present an inductive process that will in the end produce $\cC^*$. Each step $i$ of this process produces a subclass $\cC_i$ of $\cC$ for which the following hold: 
\begin{enumerate}
\item[a)] The structures in $\cC_i$ are $M$-$i$-equivalent.
\item[b)] There exists $\mA_i \in \cC_i$ and
\[
\cC_i = \{ \mB \in \cC \mid \forall j\in \ins^i(\enc(\mA_i)) \text{ the $j$th bit of $\enc(\mB)$ is $0$}\}.
\]
\end{enumerate}
Define $\cC_0 := \cC$; clearly $\cC_0$ satisfies the properties above. For $i < f(2k)$, we define $\cC_{i+1}$ to be the subclass of  $\cC_{i}$ consisting of those structures $\mA$ that on time step $i+1$ inspects an index that holds the bit $0$.\footnote{If the machine already halted on an earlier time step $t$, we stipulate that the machine inspects on time step  $i+1$ the same index that it inspected on time step $t$.}

Assume that a) and b) hold for $\cC_{i}$, we will show that the same holds for $\cC_{i+1}$. Proof of a): Let $\mA,\mB\in \cC_{i+1}$. By construction and by the induction hypothesis, $\mA$ and $\mB$ are $M$-$i$-equivalent, and on step $i+1$ $M(\enc(\mA))$ and $M(\enc(\mB))$ inspect the same index that holds $0$. Thus $\mA$ and $\mB$ are $M$-$(i+1)$-equivalent. Proof of b): It suffices to show that $\cC_{i+1}$ is nonempty; the claim then follows by construction and the property b) of $\cC_i$. By the induction hypothesis, there is a structure $\mA_i \in \cC_i$. Let $j$ be the index that $M(\enc(\mA_i))$ inspects on step $i+1$. Since $i+1\leq f(2k) < k-1$, there exists a structure $\mA_i'\in \cC_i$  such that the $j$th bit of $\enc(\mA_i')$ is $0$. Clearly $\mA'_i \in \cC_{i+1}$.

Consider the class $\cC_{k-2}$ (this will be our $\cC^*$) and $\mB\in \cC_{k-2}$ and recall that $\enc(\mB)$ is of the form $b_1\dots b_k c_1\dots c_k$. Since $\lvert \ins^{k-2}(\mB) \rvert \leq k-2$, there exists two distinct indices $i$ and $j$, $0\leq i < j\leq k-1$, such that $i,j,i+k,j+k\notin \ins^{k-2}(\enc(\mA))$. Let $\mB_{P < Q}$ denote the structure such that $\enc(\mB_{P < Q})$ is a bit string where the $i$th and $j+k$th bits are $1$ and all other bits are $0$. Similarly, let $\mB_{Q < P}$ denote the structure such that $\enc(\mB_{Q < P})$ is a bit string where the $j$th and $i+k$th bits are $1$ and all other bits are $0$. Clearly the structures $\mB_{P < Q}$ and $\mB_{Q < P}$ are in $\cC_{k-2}$ and $M$-$(k-2)$-equivalent. Since $(k-2)$ bounds above the length of computations of $M(\enc(\mB_{P < Q}))$ and $M(\enc(\mB_{Q < P}))$, it follows that the outputs of the computations are identical. This is a contradiction, for $\mB_{P < Q}$ and $\mB_{Q < P}$ are such that $M$ should accept the first and reject the second. 
\end{proof}

\begin{theorem} \label{primitive}
Let $c$ and $d$ be constant symbols in a vocabulary $\sigma$.
There does not exist an index logic formula $\varphi$ that does
not use the order predicate $\leq$ on terms of \thefirstsort{}
and that is equivalent with the formula $c\leq d$.
\end{theorem}

\begin{proof}
For the sake of a contradiction, assume that $\varphi$ is such a formula. We will derive a contradiction with Lemma \ref{lemma:order}. Without loss of generality, we may assume that the only symbols of $\sigma$ that occur in $\varphi$ are $c$ and $d$, and that $\varphi$ is a sentence (i.e., $\varphi$ has no free variables).

We define the translation $\varphi^*$ of $\varphi$ inductively. In addition to the cases below, we also have the cases where the roles of $c$ and $d$ are swapped.
\begin{itemize}
\item For $\psi$ that does not include $c$ or $d$, let $\psi^*:=\psi$.
\item For Boolean connectives and quantifiers the translation is homomorphic.
\item For $\psi$ of the form  $\IFPx{\theta}{\bar y}$, let $\psi^* := \IFPx{\theta^*}{\bar y}$.
\item For $\psi$ of the form  $c=d$, let $\psi^* := \bot$.\footnote{By $\bot$ we denote some formula that is always false, e.g, $\exists {\tt x}\, {\tt x}\neq {\tt x}$.}
\item For $\psi$ of the form $c = x$ or $x = c$, let
$\psi^* := C(x)$. 
\item For  $\psi$ of the form $x = \idx\{ \tt{x} : \theta(\tt{x}) \}$, define $\psi^*$ as
\(
x = \idx\{ {\tt x} : \theta^*( {\tt x}) \}.
\)
\item For  $\psi$ of the form $c = \idx\{ \tt{x} : \theta(\tt{x}) \}$, let
\[
\psi^* := \exists z   (z= \idx\{ {\tt x} : \theta^*( {\tt x}) \} \land C(z)),
\]
where $z$ is a fresh variable.
\end{itemize}
If $\mA$ is a $\{C,D\}$-structure such that $C^\mA$ and $D^\mA$ are disjoint singleton sets, we denote by $\mA'$ the $\{c,d\}$-structure with the same domain such that $\{c^{\mA'}\} = C^\mA$ and  $\{d^{\mA'}\} = D^\mA$. We claim that for every $\{C,D\}$-structure $\mA$  such that $C^\mA$ and $D^\mA$ are disjoint singleton sets $\{l\}$ and $\{m\}$ and every valuation $\val$ the following holds:
\[
l < m  \quad\Leftrightarrow\quad c^{\mA'} <  d^{\mA'} \quad\Leftrightarrow\quad \mA',\val \models \varphi \quad\Leftrightarrow\quad \mA,\val \models \varphi^*.
\]
This is a contradiction with Lemma \ref{lemma:order}. It suffices to proof the last equivalence as the first two are reformulations of our assumptions.
The proof is by induction on the structure of $\varphi$. The cases that do not involve the constants $c$ and $d$ are immediate. Note that by assumption, $c^\mA$ and $d^\mA$ are never equal and thus the subformula $c=d$ is equivalent to $\bot$. The case $c=x$ is also easy:
\[
\mA',\val \models c=x \quad\Leftrightarrow\quad \val(x)=c^{\mA'} \quad\Leftrightarrow\quad \val(x)\in C^\mA \quad\Leftrightarrow\quad \mA,\val \models C(x).
\]
The case for $c = \idx\{x : \theta(x)\}$ is similar:
\begin{align*}
\mA',\val \models c = \idx\{x : \theta(x)\} &\quad\Leftrightarrow\quad \mA',\val \models \exists z (z = \idx\{x : \theta(x)\} \land c=z) \\
&\quad\Leftrightarrow\quad \mA,\val \models \exists z (z = \idx\{x : \theta(x)\} \land C(z)). 
\end{align*}
All other cases are homomorphic and thus straightforward.
\end{proof}

We conclude this section by affirming that, on purely relational
vocabularies, the order predicate on \thefirstsort{} is
redundant.  The intuition for this result was given in the
beginning of this section.

\begin{theorem} \label{redundancy}
Let $\sigma$ be a vocabulary without constant or function symbols. For every sentence $\varphi$ of index logic of vocabulary $\sigma$ there exists an equivalent sentence $\varphi'$ that does not use the order predicate on terms of \thefirstsort.
\end{theorem}

\begin{proof}
We will define the translation $\varphi'$ of $\varphi$ inductively. 
Without loss of generality, we may assume that each variable that occurs in $\varphi$ is quantified exactly once (for this purpose, we stipulate that the variable $\ttx$ is quantified by the term $\idx\{ \tt{x} : \alpha(\tt{x})\}$). For every variable $x$ of \thefirstsort{} that occurs in $\varphi$, let $\alpha_x({\tt x})$ denote the unique subformula such that $\exists x (x=\idx\{{\tt x} : \alpha_x({\tt x})\} \land \psi)$ is a subformula of $\varphi$ for some $\psi$.
Note that
${\tt x}$ occurs only in $\idx\{{\tt x} : \alpha_x({\tt x})\}$. We define the following shorthands for variables ${\tt x}$ and ${\tt y}$ of \thesecondsort:
\begin{align*}
\varphi_{{\tt x} ={\tt y}} (\psi({\tt x}), \theta({\tt y})) &:= \forall {\tt z} \big( \psi({\tt z}/ {\tt x}) \leftrightarrow \theta({\tt z}/ {\tt y})\big), 
\\
\varphi_{{\tt x} < {\tt y}}  (\psi({\tt x}), \theta({\tt y})) &:=  \exists {\tt z} \Big(\big( \neg \psi({\tt z}/ {\tt x}) \land  \theta({\tt z}/ {\tt y}) \big) \land  \forall {\tt z}'  \Big({\tt z} < {\tt z}' \rightarrow  \big(  \psi({\tt z}'/ {\tt x}) \leftrightarrow \theta({\tt z}'/ {\tt y}) \big) \Big)\Big),
\end{align*}
where ${\tt z}$ and ${\tt z}'$ are fresh distinct variables of \thesecondsort. In the formulas above,  $\psi({\tt z}/ {\tt x})$ denotes the formula that is obtained from $\psi$ by substituting each free occurrence of $\ttx$ in $\psi$ by $\ttz$.
The translation $\varphi\mapsto \varphi'$ is defined as follows:
\begin{itemize}
\item For formulae that do not include variables of \thefirstsort{}, the translation is the identity.
\item For Boolean connectives and quantifiers of \thesecondsort{}, the translation is homomorphic.
\item For $\psi$ of the form  $\IFPx{\theta}{\bar y}$, let $\psi' := \IFPx{\theta'}{\bar y}$.
\item For $\psi$ of the form $x\leq y$, let  $\psi' := \varphi_{{\tt x} ={\tt y}}(\alpha_x({\tt x}), \alpha_y({\tt y})) \lor  \varphi_{{\tt x} < {\tt y}}(\alpha_x({\tt x}), \alpha_y({\tt y}))$.
\item For  $\psi$ of the form $x = \idx\{ {\tt y} : \theta({ \tt y}) \}$, define $\psi' := \varphi_{{\tt x} ={\tt y}}(\alpha_x({\tt x}), \theta({\tt y}))$.
\item For  $\psi$ of the form $\exists x (x = \idx\{ \tt{x} : \alpha(\tt{x}) \land \theta \}$, define $\psi' := \theta'$.
\end{itemize}
By a straightforward inductive argument it can be verified
that the translation preserves equivalence.
\end{proof}

\section{Index logic with partial fixed points}{\label{sec:nifpplog}}

 In this section, we introduce a variant of index logic defined in Section~\ref{sec:ifpplog}. This logic, which we denote as IL(PFP), is defined by simply replacing the inflationary fixed point operator IFP in the definition of index logic by the partial fixed point operator PFP. We stick to the standard semantics of the PFP operator.
We define that 
\[ 
\mathbf{A},\mathit{val} \models [\mathrm{PFP}_{\bar{\mathtt{x}}, X} \varphi]\bar{\tty} \text{ iff }  \mathit{val}(\bar{\tty}) \in \mathrm{pfp}(F^{\bf A, \it{val}}_{\varphi, \bar{\mathtt{x}}, X}),
\]  
where $\mathrm{pfp}(F^{\bf A, \it{val}}_{\varphi, \bar{\mathtt{x}}, X})$ denotes the \emph{partial} fixed point of the operator $F^{\bf A, \it{val}}_{\varphi, \bar{\mathtt{x}}, X}$ (see the description above Definition~\ref{semanticsIndexLogic}).
The \emph{partial fixed point} $\mathrm{pfp}(F)$ of an operator $F : {\cal P}(B) \rightarrow {\cal P}(B)$ is defined as the fixed point of $F$ obtained from the sequence $(S^i)_{i\in\mathbb{N}}$, where $S^0 := \emptyset$ and $S^{i+1} := F(S^i)$, if such a fixed point exists. If such a fixed point does not exist, then  $\mathrm{pfp}(F) := \emptyset$.
%

It is well known that first-order logic extended with partial fixed point operators captures $\pspace$.
As a counterpart for this result, we will show that IL(PFP) captures the complexity class polylogarithmic space ($\polylogspace$).
Recall that in IL(PFP) the relation variables bounded by the PFP operators range over (tuples of) $\Num(\mA)$, where $\mA$ is the interpreting structure. Thus, the maximum number of iterations before reaching a fixed point (or concluding that it does not exist), is \emph{not} exponential in the size $n$ of $\mA$, as in FO(PFP).
Instead, it is \emph{quasi-polynomial}, i.e., of size $O(2^{\log^k n})$, for some constant $k$.
This observation is, in part, the reason why IL(PFP) characterizes $\polylogspace$.
Finally, by an analogous argument that proves the well-known relationship $\mathrm{PSPACE} \subseteq \mathrm{DTIME}(2^{n^{O(1)}})$, it follows that $\polylogspace \subseteq \mathrm{DTIME}(2^{\log^{O(1)} n})$.

\subsection{The Complexity Class $\polylogspace$}

Let $L(M)$ denote the class of structures of a given signature $\sigma$ accepted by
 a direct-access Turing machine $M$. We say that $L(M) \in \mathrm{DSPACE}[f(n)]$ if $M$ visits at
most $O(f(n))$ cells in each work-tape before accepting or rejecting an input
structure whose domain is of size $n$.
We define the class of all languages decidable by a deterministic
direct-access Turing machines in \emph{polylogarithmic space} as
follows:
\begin{equation*}
\polylogspace := \bigcup_{k \in \mathbb{N}}
\mathrm{DSPACE}[(\left\lceil\log n \right\rceil)^k].
\end{equation*}
Note that it is equivalent whether we define the class $\mathrm{PolylogSpace}$ by means of direct-access Turing machines or random-access Turing machines. Indeed, by Theorem~\ref{directrandom} and by the fact that the (standard) binary encoding of a structure $\mA$ is of size polynomial with respect to the cardinality of its domain $A$, the following corollary is immediate.

\begin{corollary}	\label{tba}
A class of finite ordered structures $\cal C$ of some fixed vocabulary $\sigma$ is decidable by a random-access Turing machine working in $\polylogspace$ with respect to $\hat{n}$,
where $\hat{n}$ is the size of the binary encoding of the input structure, iff $\cal C$ is decidable by a direct-access Turing machine in $\polylogspace$ with respect to $n$, where $n$ is the size of the domain of the input structure.
\end{corollary}

Moreover, in the context of $\polylogspace$, there is no need for random-access address-tape for the input; $\polylogspace$ defined with random-access Turing machines coincide with $\polylogspace$ defined with (standard) Turing machines that have sequential access to the input.

\begin{proposition}
A class of finite ordered structures $\cal C$ of some fixed vocabulary $\sigma$ is decidable by a random-access Turing machine working in $\polylogspace$ with respect to $\hat{n}$ iff $\cal C$ is decidable by a standard (sequential-access) Turing machine in $\polylogspace$ with respect to  $\hat{n}$, where $\hat{n}$ is the size of the binary encoding of the input structure.
\end{proposition}
\begin{proof}
We give the idea behind the proof; the proof itself is straightforward. We take as the definition of the standard (sequential-access) Turing machine the definition of the random-access Turing machine given in Section \ref{sec:random-access}, except that we suppose a sequential-access read-only-head for the input tape, and remove the address-tape.

A random-access Turing machine $M_r$ can simulate a sequential-access Turing machine $M_s$ directly by using its address-tape to simulate the movement of the head of the sequential-access input-tape. In the simulation, when the head of the input-tape of $M_s$ is on the $i+1$-th cell, the address-tape of $M_r$ holds the number $i$ in binary, and hence refers to the $i+1$-th cell of the input. When the head of the input-tape of $M_s$ moves right, the machine $M_r$ will increase the binary number in its address-tape by one. Similarly, when the head of the input-tape of $M_s$ moves left, the machine $M_r$ will decrease the binary number in its address-tape by one. A total of $\lceil \log n \rceil$ bits suffices to access any bit of an input of length $n$.
Clearly increasing or decreasing a binary number of length at most $\lceil \log n \rceil$ by one can be done in $\polylogspace$. The rest of the simulation is straightforward.

The simulation of the other direction is a bit more complicated, as after each time the content of the address-tape of the  random-access machine is updated, we need to calculate the corresponding position of the head of the input-tape of the sequential-access machine. However, this computation can be clearly done in $\polylogspace$: We use a work-tape of the sequential-access machine to mimic the address-tape of the sequential-access machine, and an additional work-tape as a binary counter. After each computation step of the random-access machine, the sequential-access machine moves the head of its input tape to its leftmost cell, formats the work-tape working as a binary counter to contain exactly the binary number that is written on the address-tape. Then the sequential-access machine moves the head of its input-tape right step-by-step simultaneously decreasing the binary counter by $1$. Once the binary counter reaches $0$, the head of the input tape is in correct position. The rest of the simulation is straightforward. 
\end{proof}

Since the function $\left\lceil\log n \right\rceil$ is \emph{space constructible} (s.c. for short) (see~\cite{papadimitriou}, where these functions are denoted as \emph{proper}), and for any two s.c. functions their product is also s.c., we get that for any $k \geq 1$ the function $(\left\lceil\log n \right\rceil)^k$ is s.c. Hence, by Savitch's theorem, we obtain the following result.

\begin{fact}
For any $k \geq 1$, it holds that $\mathrm{NSPACE}[(\left\lceil\log n \right\rceil)^k] \subseteq \mathrm{DSPACE}[(\left\lceil\log n \right\rceil)^{2k}]$. Thus, nondeterministic and deterministic $\polylogspace$ coincide. 
\end{fact}

\subsection{Index logic with partial fixed point operators captures $\polylogspace$}
To encode a configuration of polylogarithmic size, we follow a similar strategy as in Theorem \ref{captureResult}, i.e., in the proof of the characterization of $\dpolylog$ by $\mathrm{IL(IFP)}$.
The difference here is that there is no reason to encode the whole history of a computation in the fixed point. At a time step $t$ it suffices that the configuration of the machine at time step $t-1$ is encoded; hence, we may drop the variables $\bar{t}$, from the fixed point formula defined on page \pageref{IFPformula}. Moreover, we make a small alteration to the Turing machines so that acceptance on an input structure will correspond to the existence of a partial fixed point.

\begin{theorem}\label{captureResultPolyLogSpace}
Over ordered finite structures, $\mathrm{IL(PFP)}$ captures $\polylogspace$.
\end{theorem}

\begin{proof}
The direction of the proof that argues that IL(PFP) can indeed be evaluated in $\polylogspace$ is straightforward. Let $\psi$ be an IL(PFP)-sentence, we only need to show that there exists a direct-access Turing machine $M_{\psi}$ working in $O(\log^{d} n)$ space, for some constant $d$, such that for every structure $\mathbf{A}$ and valuation $\it{val}$, it holds that $\mathbf{A} \in L(M_{\psi})$ iff $\mathbf{A}, \it{val} \models \psi$.
Note that, in an induction on the structure of $\psi$, all the cases, except the case for the $\mathrm{PFP}$ operator, are as in the proof of Theorem~\ref{captureResult}. Clearly if a formula can be evaluated in $\mathrm{PolylogTime}$ it can also be evaluated in $\mathrm{PolylogSpace}$. For the case of the $\mathrm{PFP}$ operator (using a similar strategy as in~\cite{ef_fmt2}), we set a counter to
$2^{\log^{r} n}$, using exactly $\log^{r} n$ cells in a work-tape, where $r$ is the arity of the relation variable $X$ bounded by the $\mathrm{PFP}$ operator.
To evaluate the $\mathrm{PFP}$ operator, say on a formula $\varphi (\bar{\mathtt{x}}, X)$, $M$ will iterate evaluating $\varphi$, decreasing the counter in each iteration. When the counter gets to $0$, $M$ checks whether the contents of the relation $X$ is equal to its contents in the following cycle, and whether the tuple given in the  $\mathrm{PFP}$ application belongs to it. If both answers are positive, then $M$ accepts, otherwise, it rejects. This suffices to find the fixed point (or to conclude that it does not exist), as there are $2^{\log^{r} n}$ many relations of arity $r$ with domain $\{0,\dots, \lceil \log n \rceil-1\}$.


For the converse, let $M = (Q, \Sigma, \delta, q_0, F, \sigma)$ be an $m$-tape direct-access Turing machine that works in $\polylogspace$. As in the proof of Theorem \ref{captureResult}, we assume w.l.o.g., that $F = \{q_a\}$ (i.e., there is only one accepting state), $|Q| = a+1$, and $Q = \{q_0, q_1, \ldots, q_{a}\}$. In addition to the assumptions made in the proof of Theorem \ref{captureResult}, 
we assume that once the machine reaches an accepting state, it will not change its configuration any longer; that is, all of its heads stay put, and write the same symbol as the head reads.
Note that the machine $M$ accepts if and only if $M$ is in the same accepting configuration during two consecutive time steps.

We build an IL(PFP)-sentence $\psi_{M}$ such that for every structure $\mathbf{A}$ and valuation $\it{val}$, it holds that $\mathbf{A} \in L(M)$ iff $\mathbf{A},\it{val} \models \psi_{M}$. The formula is a derivative of that of Theorem \ref{captureResult} and is defined using a simultaneous PFP operator. In the formula below,  $S_{q_0}, \ldots, S_{q_{a}}$ denote $0$-ary relation variables that range over the values \emph{true} and \emph{false}. We define
\begin{equation*}\label{PFPformula}
 \psi_{M} := [\textrm{S-PFP}_{
S_{q_a}, \mathrm{A}, \mathrm{B}_1, \mathrm{B}_2, \mathrm{B}_3, \mathrm{C} 
} 
\;\varphi_{q_a}, \Phi_\mathrm{A}, \Phi_{\mathrm{B}_1}, \Phi_{\mathrm{B}_2}, \Phi_{\mathrm{B}_3}, \Phi_{\mathrm{C}}],
\end{equation*}   
where
\[\mathrm{A} = S_{q_0}, \ldots, S_{q_{a-1}} \quad
 \mathrm{B}_1 = \bar{p}, T^0_1, \ldots, \bar{p}, T^0_m \quad
\mathrm{B}_2 = \bar{p}, T^1_1, \ldots, \bar{p}, T^1_m
\]
\[\mathrm{B}_3 = \bar{p}, T^\sqcup_1, \ldots, \bar{p}, T^\sqcup_m \quad
\mathrm{C} = \bar{p}, H_1, \ldots, \bar{p}, H_m
\]
\[\Phi_\mathrm{A} = \varphi_{q_0}, \ldots, \varphi_{q_{a-1}} \quad
\Phi_{\mathrm{B}_1} = \psi_{01}, \ldots, \psi_{0m} \quad 
\Phi_{\mathrm{B}_2} = \psi_{11}, \ldots, \psi_{1m}\]
\[\Phi_{\mathrm{B}_3} = \psi_{\sqcup1}, \ldots, \psi_{\sqcup m}\quad
\Phi_{\mathrm{C}} = \gamma_{1}, \ldots, \gamma_{m}.\]

The formulae used in the PFP operator are defined in the same way as in Theorem  \ref{captureResult}; with the following two exceptions.
\begin{enumerate}
\item The formulae of the form $\alpha^0_i(\bar{p}, \bar{t} - 1)$ are replaced with the analogous formulae $\alpha^0_i(\bar{p})$ obtained, by simply removing the variables referring to time steps.
\item Subformulas of the form $\bar{t}\sim 0$ are replaces with $\neg S_{q_0} \land \ldots \land \neg S_{q_{a-1}}$, which will be true only on the first iteration of the fixed point calculation.
\end{enumerate}

Following the proof of Theorem  \ref{captureResult}, it is now easy to show that  $\mA,\it{val} \models \psi_{M}$
if and only if \emph{$M$ accepts $\mA$}.
\end{proof}

\section{Discussion}

An interesting open question concerns order-invariant queries.
Indeed, while index logic is defined to work on ordered
structures, it is natural to try to understand which queries
about ordered structures that are actually invariant of the
order, are computable in PolylogTime.  Results of the kind given
by Proposition~\ref{emptiness} already suggest that very little
may be possible.  Then again, any polynomial-time numerical
property of the size of the domain is clearly computable.  We
would love to have a logical characterization of the
order-invariant queries computable in PolylogTime.

Another natural direction is to get rid of Turing machines
altogether and work with a RAM model working directly on
structures, as proposed by Grandjean and Olive
\cite{grol_graph_lin}.  Plausibly by restricting their model to
numbers bounded in value by a polynomial in $n$ (the size of the
structure), we would get an equivalent PolylogTime complexity
notion.

In this vein, we would like to note that extending index logic
with numeric variables that can hold values up to a polynomial in
$n$, with arbitrary polynomial-time functions on these, would
be useful syntactic sugar that would, however, not increase the
expressive power.

\section*{References}

\bibliography{database}

\end{document}